\documentclass[ba]{imsart}
\pubyear{2022}
\volume{TBA}
\issue{TBA}
\firstpage{1}
\lastpage{1}

\usepackage{amsthm}
\usepackage{amsmath}
\usepackage{natbib}
\usepackage[colorlinks,citecolor=blue,urlcolor=blue,filecolor=blue,backref=page]{hyperref}
\usepackage{graphicx}

\startlocaldefs

\usepackage{multicol}

\usepackage{tikz}
\usetikzlibrary{cd} 
\usetikzlibrary{backgrounds}

\usetikzlibrary{fit,positioning,arrows.meta,bayesnet,shapes.geometric,chains,matrix,scopes,calc,decorations.markings}
\tikzstyle{param}=[circle, minimum size = 0.7cm, thick, draw=black!100, fill = gray!10, node distance = 0.5cm]
\tikzstyle{data}=[rectangle, minimum size = 0.7cm, thick, draw =black!100, node distance = 0.5cm]
\tikzstyle{model}=[rectangle, minimum size = 1cm, thick, draw=black!100, node distance = 0.5cm]

\numberwithin{equation}{section}
\theoremstyle{plain}

\newtheorem{Theorem}{Theorem}[section]
\newtheorem{corollary}{Corollary}[Theorem]

\newtheorem{proposition}{Proposition}[Theorem]
\newtheorem{Definition}{Definition}

\usepackage{amsfonts} 
\def\R{\mathbb{R}}
\def\I{\mathbb{I}}
\def\Y{{\tilde Y}}
\def\y{{\tilde y}}
\def\l{{l}}
\def\p{{\tilde p}}
\def\q{{\tilde q}}

\def\b{\mathcal{V}}
\newcommand\widebar[1]{\bar{#1}}%

\newcommand{\crvBr}[1]{\left\{#1\right\}}
\newcommand{\sqrBr}[1]{\left[#1\right]}

\endlocaldefs

\begin{document}

\begin{frontmatter}
\title{Valid belief updates for prequentially additive loss functions arising in Semi-Modular Inference}
\runtitle{Valid variants of Semi-Modular Inference}

\begin{aug}
\author{\fnms{Geoff K.} \snm{Nicholls}\thanksref{addr1}\ead[label=e1]{nicholls@stats.ox.ac.uk}},
\author{\fnms{Jeong Eun} \snm{Lee}\thanksref{addr2}\ead[label=e2]{kate.lee@auckland.ac.nz}},
\author{\fnms{Chieh-Hsi} \snm{Wu}\thanksref{addr3}%
\ead[label=e3]{C-H.Wu@soton.ac.uk}}
\and
\author{\fnms{Chris U.} \snm{Carmona}\thanksref{addr1}\ead[label=e4]{carmona@stats.ox.ac.uk}}

\runauthor{G.K. Nicholls et al.}

\address[addr1]{
Department of Statistics,
  University of Oxford,
  Oxford, UK.
    \printead{e1} 
    \printead{e4} 
}
\address[addr2]{
Department of Statistics,
  University of Auckland,
  Auckland, NZ.
    \printead{e2}
}
\address[addr3]{
Mathematical Sciences,
  University of Southampton,
  Southampton, UK.
    \printead{e3}
}
\end{aug}

\begin{abstract}
Model-based Bayesian evidence combination leads to models with multiple parameteric modules. In this setting the effects of model misspecification in one of the modules may in some cases be ameliorated by cutting the flow of information from the misspecified module. Semi-Modular Inference (SMI) is a framework allowing partial cuts which modulate but do not completely cut the flow of information between modules. We show that SMI is part of a family of inference procedures which implement partial cuts. It has been shown that additive losses determine an optimal, valid and order-coherent belief update. The losses which arise in Cut models and SMI are not additive. However, like the prequential score function, they have a kind of prequential additivity which we define. We show that prequential additivity is sufficient to determine the optimal valid and order-coherent belief update and that this belief update coincides with the belief update in each of our SMI schemes. 
\end{abstract}



\end{frontmatter}






\section{Introduction}

Bayesian analysis integrates different sources of information or ``modules" into a single analysis through Bayes theorem and quantifies uncertainties in parameters. The information in each module, which may be prior belief or observations or both, is encoded as a parametric model. Evidence synthesis can give better predictability, more precise estimation, and access to shared parameter estimation \citep{Ades2006, Sweeting2009, Harris2012,Fithian2015,Pacifici2017,Nicholson2021covid}.  
 
As modules are added to an overall model, there is an increasing hazard for misspecification. Methods that help us carry out Bayesian analysis on misspecified models have been in development for some time   without explicit consideration of modularisation. 
We divide these into three classes. Methods which temper the likelihood lead to power posteriors, \cite{Walker2001,Grunwald2012,Miller2018a}, methods which use bootstrapping in a Bayesian setting, including Weighted likelihood Bootstrap \citep{Newton1991,Newton1994,Lyddon2019}, the Posterior Bootstrap \citep{Pompe2021} and BayesBag \citep{Buhlmann2014,Huggins2021}, and methods which replace the likelihood with some more general loss function mediating data and parameter, including PAC-Bayes \citep{Germain2016,Zhang2006,McAllester1998,Shawe1997}, Gibbs posteriors \citep{Zhang2006, Jiang2008} and Generalized Bayes \citep{Bissiri2016, Grunwald2017}
are relevant in the multi-modular setting.

Multi-modular Bayesian inference with misspecified models has some features which distinguish it from misspecification in single module settings. \citet{Liu2009} gave an early ``modularization" analysis. Markov melding \citep{Goudie2019melding} and Bayesian melding \citep{Poole2000} can be characterised as dealing with priors which conflict across modules. In our own work
we assume that modules have been identified as either misspecified or well-specified. This may be the conclusion of a first stage Bayesian analysis of the overall multi-modular model. A typical and well-founded objective is to estimate the parameters of a well-specified module making careful use of information from misspecified modules. 

{\it Cut-model} inference \citep{Plummer2015} has proven very effective in this setting. We discuss this in detail below. It can be thought of as a kind of sequential imputation procedure, in which the distribution of a shared parameter is imputed form the information in one module and then passed on as a kind of prior for the shared parameter in a second module. This is not Bayesian inference, as information from the second module does not inform the shared parameter. A early form of Cut model inference has been available in WinBUGS \citep{Spiegelhalter2014} for some time.
Cut models have found many applications: air pollution \citep{BLANGIARDO2011379}, epidemiological models \citep{ Maucort2008, Finucane2016, Li2017, Nicholson2021covid, Teh2021covid}, meta-analysis \citep{Lunn2009, Lunn2013, Kaizar2015} and propensity scores \citep{Zigler2013, Zigler2014}. \citet{Jacob2017b} gives an overview of modularized Bayesian analysis including Cut-models from the perspective of statistical decision theory and \citet{Pompe2021} gives asymptotic properties. Nested MCMC \citep{Plummer2015} is commonly used for fitting Cut models. New developments include variational approximation \citep{Yu2021variationalcut} and a computationally efficient variant of nested-MCMC \citep{Liu2020sacut}.

In Cut-model inference, feedback from the suspect module is completely cut. However, there may be a bias-variance trade-off: if the parameters of a well-specified module are poorly informed by ``local" information then limited information from misspecified modules may allow us to bring the uncertainty down without introducing significant bias.   
Semi-Modular Inference ($\eta$-SMI, \cite{carmona20}) generalises Cut-model inference as it offers a means by which we can control the influence of suspect modules on the fit for a good module. Candidate posterior distributions are indexed by an associated influence parameter $\eta$. At $\eta=1$ $\eta$-SMI is standard Bayesian inference and at $\eta=0$ $\eta$-SMI reproduces the Cut-model. \citet{carmona20} suggest choosing $\eta$ maximizing the expected log pointwise predictive density (ELPD) though this choice is not an essential part of their method and other criteria \citep{Wu2020calib-comparison} may be more appropriate in different settings.  \cite{Liu2021generalized} adapt $\eta$-SMI for Geographically Weighted Regression using an influence parameter across likelihood factors which is modeled as a function of distance between the spatial observation locations.


Many of the papers cited up to this point propose probability distributions which can be seen as alternative posteriors, in the sense that they offer quantification of uncertainty. \cite{Bissiri2016} call these alternative data-informed mappings ``belief updates" and characterise the optimal, valid and order-coherent belief update as a Gibbs posterior in a Generalised Bayes setting. Several existing belief updates, such as the power posterior, are known to be valid. However the characterisation of a valid belief update holds for losses which are additive. We extend the class of losses for which valid belief updates can be defined. In particular we show the Cut-models and $\eta$-SMI have a kind of prequential additivity which is sufficient for the theory of \cite{Bissiri2016} to apply. We point out alternative SMI procedures, which we call $\delta$-SMI and $\gamma$-SMI. These offer different interpolating sequences of candidate posterior distributions. The $\delta$-SMI sequence progressively ``blurs" the data away using a procedure resembling Approximate Bayesian Computation (ABC) but otherwise offers candidate posterior distributions which are often similar to those of $\eta$-SMI. \citet{Liu2021generalized} consider a parallel framework for deriving belief updates due to \cite{Zhang2006}. This resembles PAC-Bayesian approaches \citep{McAllester1998, Shawe1997} and is distinct from the approach of \cite{Bissiri2016}.

This paper has three main parts. In the first part (Section~\ref{sec:intro}) we show how prequential additivity leads to valid updates of belief and show that the Cut model is a valid belief update. In the second part (Sections~\ref{sec:smi} and \ref{sec:properties_of_SMI}) we introduce some SMI-variants and consider their properties. We use the theory from the first part to show they are valid belief updates. In the third part (Section~\ref{sec:data_analyses}) we give some simple examples to explore the behavior of one of the SMI-variants introduced in the second part.

\section{Belief updates} \label{sec:intro}


\subsection{Belief updates and the Cut model}\label{sec:beliefupdate-cutmodel-intro}

\begin{figure}[!ht]
  \begin{center}
    \begin{tikzpicture}
      \node (Z) [data] {$Z$};
      \node (Y) [data, right=of Z, xshift=0.5cm] {$Y$};
      \node (phi) [param, below=of Z] {$\varphi$};
      \node (theta) [param, below=of Y] {$\theta$};
      \edge {phi} {Y, Z};
      \edge {theta} {Y};
      \draw[dashed,red,line width=0.5mm] (0.85,0) to (0.85,-1.25);
      \node[text width=3cm] at (0.5,1) {Module 1};
      \node[text width=3cm] at (2.75,1) {Module 2};
    \end{tikzpicture}
  \end{center}
  \caption{Graphical representation of a simple multi-modular model. The Bayes posterior for this model in given in \eqref{eq:bayespost}. The dashed vertical line indicates that a Cut model is used. The Cut-model posterior is given in \eqref{eq:cutpost}.}
  \label{fig:toy_multimodular_model}
\end{figure}
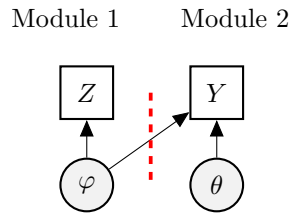

Consider the two-module configuration of Fig.~\ref{fig:toy_multimodular_model}. Let $Z=(Z_1,...,Z_m)$ and $Y=(Y_1,...,Y_n)$ be two vectors of data with model parameter vectors $\varphi$ and $\theta$. In our notation below we take the sample spaces to be $Z_j\in \R^{d_z},\ j=1,...,m$, $Y_i\in \R^{d_y},\ i=1,...,n$, $\varphi\in \R^{p_\varphi}$ and $\theta\in \R^{p_\theta}$ though this is not an essential restriction and our final example takes discrete data.
In Generalised Bayesian inference with a Gibbs posterior \citep{Chernozhukov2003,Zhang2006,Jiang2008,Bissiri2016} we have a prior $\pi_0(\theta,\varphi)$ (a density here) and a loss $l(\varphi,\theta;Y,Z)$ connecting the data and parameters, measuring how well the parameters agree with the data.

A {\it belief update} $\psi$ is a rule which updates the prior distribution, taking into account the data through the loss. It determines an updated belief distribution $\p$. When we are choosing between different belief updates we refer to these as ``candidate posteriors".
  For the model in Fig.~\ref{fig:toy_multimodular_model}, we write
  \begin{equation}\label{eq:belief_update_first}
    \p(\varphi, \theta \mid Y,Z) = \psi\{ l(\varphi, \theta;Y,Z), \pi_0 \}.
  \end{equation}
When we specify a probability distribution from a loss, via a belief update, we write it as if it is a conditional probability density. We use $p()$ (for densities over data) and $\pi()$ (for densities over parameters) when they may be understood as a conditional probability. However, many belief updates, like the Cut model below, do not yield conditional probability distributions. We write $\p()$ for probability distributions of this kind.

In Generalised Bayes the belief update from the prior to the posterior is
\begin{equation}\label{eq:genbayes_expost_firsttime}
\p(\varphi, \theta|Y,Z)\propto \exp(-l(\varphi, \theta;Y,Z)) \pi_0(\varphi, \theta).
\end{equation}
For example, in Bayesian inference with observation models $p(Z | \varphi)$ and $p(Y | \varphi,\theta)$ (probability densities, say) the posterior distribution of $(\varphi,\theta)$ is
\begin{equation} \label{eq:bayespost}
  \pi(\varphi,\theta \mid Y, Z) \propto p(Z\mid \varphi) p(Y\mid \varphi,\theta)\pi(\varphi,\theta),
\end{equation}
and so the loss in \eqref{eq:genbayes_expost_firsttime} ``must have been'' the negative log-likelihood,
\begin{equation} \label{eq:bayes_loss}
  \l^{(b)}( \varphi,\theta ; Y, Z) = - \log p(Z \mid \varphi) -\log p(Y \mid \varphi,\theta).
\end{equation}
In this paper we follow \cite{Bissiri2016} and ask why \eqref{eq:genbayes_expost_firsttime} is a valid belief update of $\pi_0$ in the context of Cut-model inference \citep{Plummer2015} and in related forms of Semi-Modular Inference (SMI, \cite{carmona20}). 

We identify a feature of this setup which does not seem to have been considered explicitly to date: the loss itself may depend on our state of knowledge of the relation between parameters, for example, on $\pi_0(\theta|\varphi)$. This is present in the Gibbs posterior for the Cut model.
The Cut ``posterior'' for  $(\varphi,\theta)$ is \citep{Plummer2015}
\begin{equation} \label{eq:cutpost}
  \p^{(c)}(\varphi,\theta \mid Y, Z) = \pi(\varphi \mid Z) \pi(\theta \mid Y,\varphi).
\end{equation}
The Cut is indicated in Figure~\ref{fig:toy_multimodular_model} by the vertical dashed line.
It is called a ``Cut model" because the flow of information from the $Y$-module into the $Z$-module has been cut. The flow of information between modules of a Cut model is asymmetrical. This makes sense when the generative model for the data $Y$ is misspecified, but the generative model for the $Z$-module is correct.
The idea is to infer or ``impute" $\varphi$ using the reliable $Z$-model and stop misspecification in the $Y$-model from biasing that analysis. 

The Cut model above can be written in terms of the Bayes posterior
\begin{equation} \label{eq:cut_from_post}
  \p^{(c)}(\varphi,\theta \mid Y, Z) \propto  \pi(\varphi,\theta \mid Y, Z)/p(Y|\varphi)
\end{equation}
where
\begin{equation}\label{eq:marginal_Y_phi}
  p(Y|\varphi)=\int p(Y|\varphi, \theta)\pi_0(\theta|\varphi)d\theta.
\end{equation}
In our setting the likelihoods $p(Y|varphi,\theta)$ and $p(Z|\varphi)$ can be easily evaluated, but $p(Y|\varphi)$ cannot.
If the Cut model is a belief update with a Gibbs posterior, then the loss in \eqref{eq:genbayes_expost_firsttime} yielding \eqref{eq:cutpost} ``must have been''
\begin{equation} \label{eq:cut_loss}
  \l^{(c)}( \varphi,\theta ; Y, Z , \pi_0) = \l^{(b)}( \varphi,\theta ; Y, Z ) + \log p(Y \mid \varphi).
\end{equation}
 The loss function for the Cut model depends on the prior $\pi_0(\theta|\varphi)$ through the term $ \log p(Y \mid \varphi)$. 
In the setting of \cite{bissiriwalker12a} this prior dependence could be thought of as another ``piece of information'' alongside $Y,Z$. They write the loss $l(\xi;I)$ where $\xi=(\varphi,\theta)$ is the parameter and $I=(Y,Z)$ is the data or ``information'' informing the parameters. We recover that setup, at least formally, if we write $I=(Y,Z,\pi_0)$.
However, \cite{Bissiri2016} determine valid belief updates for {\it additive} losses only (see below), and we will see that the Cut-model loss is not additive, so we can ask, what is the valid belief update for the Cut-model loss? Does it coincide with the Cut-model posterior?

\subsection{Additive losses}\label{sec:additivity}

We now consider how a loss might be additive in this setting. Consider conditionally iid data $Y_1,...,Y_n|\varphi,\theta$ and $Z_1,...,Z_m|\varphi$. Let
$Y^{(1:K)}=(Y^{(1)},...,Y^{(K)})$ and $Z^{(1:K)}=(Z^{(1)},...,Z^{(K)})$
be  partitions of $Y$ and $Z$ into $K$ sets, which may be empty, with the data taken in any order. 
\begin{Definition}\label{def:additivity_first} (Additivity)
Loss function $l(\varphi, \theta;Y,Z)$ is additive if
\[
l(\varphi, \theta;Y,Z)=\sum_{k=1}^K l(\varphi, \theta;Y^{(k)},Z^{(k)})
\]
for any partition $Y^{(1:K)},Z^{(1:K)}$ of the data $(Y,Z)$.
\end{Definition}

The Bayes loss $l^{(b)}$ in \eqref{eq:bayes_loss} is additive for iid data. In contrast, the Cut loss in \eqref{eq:cut_loss} is not additive over $k$ as it depends on the marginal $p(Y|\varphi)$ in \eqref{eq:marginal_Y_phi}. However, the Cut loss depends on the evolving state of knowledge and this needs to be accounted for in the accumulated loss. The ``prequential score'' for prediction (\cite{dawid15}, Section 4) has a similar dependence on an evolving predictive distribution and so we call this property {\it prequential additivity}. 

\begin{Definition}\label{def:preq_add_first} (Prequential Additivity)
Let a belief update $\psi^{(q)}$ be given. For $k=1,...,K$ let 
\begin{equation}\label{eq:q-in-def-preq-add}
\q_k(\varphi,\theta)=\psi^{(q)}(l(\varphi, \theta;Y^{(1:k)},Z^{(1:k)},\pi_0),\pi_0)
\end{equation}
be the belief distribution for $\varphi$ and $\theta$ after the arrival of the first $k$ sets of data $Y^{(1:k)},Z^{(1:k)}$. Let $\q_0=\pi_0$.
The loss $l$ is {\it prequentially additive with respect to the belief update} $\psi^{(q)}$ if the total accumulated loss over a sequence of measurements $(Y^{(k)},Z^{(k)}), k=1,...,K$ is equal to the loss from a single bulk measurement,
\begin{equation}\label{eq:intro_preq_additivity}
   l(\varphi, \theta;Y,Z,\pi_0)=\sum_{k=1}^K l(\varphi, \theta;Y^{(k)},Z^{(k)},\q_{k-1}),
\end{equation}
for any partition $Y^{(1:K)},Z^{(1:K)}$ of the data $(Y,Z)$. 
\end{Definition}
This is a condition on a predefined loss, not the definition of the total loss as is the case for the prequential score, so it will only hold if there is a relation between the loss $l$ and the belief update $\psi^{(q)}$. 
A loss which does not depend on the prior and is additive is clearly prequentially additive. However, for example, the Cut-model loss is prequentially additive but not additive.

\begin{proposition}\label{prop:cutmodeladditive}
  The Cut-model loss $l^{(c)}( \varphi,\theta ; Y, Z ,\pi_0)$ in \eqref{eq:cut_loss} is prequentially additive with respect to the belief update 
  \[
  \psi^{(q)}(l^{(c)}(\varphi, \theta;Y,Z,\pi_0),\pi_0)\propto \exp(-l^{(c)}(\varphi, \theta;Y,Z,\pi_0)) \pi_0(\varphi, \theta).
 \]
 which is just the Cut-model posterior.
\end{proposition}
\begin{proof}
see Appendix~\ref{sec:prop:cutmodeladditive:proof}. The proof is closely related to the proof of order-coherence of Cut-model inference given in \cite{carmona20}.
\end{proof}






\subsection{Order-coherence and valid belief updates}
\label{sec:preq-add-preq-co}

In this section we show that the conclusions of \cite{Bissiri2016} extend to cover prequentially additive losses. We need this extension because prequential additivity is a weaker condition than the assumed additivity. We begin by defining order-coherence. 

Consider a general partition of the data $(Y,Z)$ into $K=2$ arbitrary subsets, as in the previous section, with $Y^{(1:2)}=(Y^{(1)},Y^{(2)})$ and $Z^{(1:2)}=(Z^{(1)},Z^{(2)})$. 
A belief update $\psi$ is \emph{order-coherent} in the sense of \cite{Bissiri2016} if the posterior for independent data is the same regardless of whether we update belief from the prior, taking all the data in one tranche, or update with $Y^{(1)},Z^{(1)}$ and use the resulting posterior as the prior for a belief update with $Y^{(2)},Z^{(2)}$. In our setting with prior-dependence in the loss function we have the following definition.
\begin{Definition} \label{def:preq-order-coherent} (order-coherence)
  Let a belief update $\psi^{(q)}$ be given and let
  \begin{equation}\label{eq:p1-in-def-order-co}
    \q_1(\varphi, \theta)=\psi^{(q)}\{ l(\varphi, \theta;Y^{(1)},Z^{(1)},\pi_0),\pi_0\}.
  \end{equation}
  Belief update $\psi^{(q)}$ is order-coherent if
  \begin{equation}\label{eq:preq-order-coherent}
    \psi^{(q)}\{ l(\varphi,\theta;Y,Z,\pi_0), \pi_0 \}=\psi^{(q)}\{ l(\varphi,\theta;Y^{(2)},Z^{(2)},\q_1) , \q_1 \},
  \end{equation}
  for every $n,m>0$ and every partition of the data taken in any order.
\end{Definition}
The property is defined for $K=2$ as it will hold for sequential belief updates along partitions of the data into $K>2$ sets if it holds for $K=2$. Order-coherence seems to us an axiomatic property for a belief update.


\cite{Bissiri2016} show that the optimal, valid and order-coherent
belief update $\psi$ is the probability measure $\nu(d\theta\,d\varphi)$ minimising the loss
\begin{equation}\label{eq:big_post_loss}
  L(\nu; Y,Z,\pi_0)=\int l(\varphi, \theta;Y,Z,\pi_0) \nu(d\theta\, d\varphi)+KL(\nu||\pi_0)
\end{equation}
over $\nu\in \mathcal F$ where $\mathcal F$ is the family measures, absolutely continuous with respect to the measure of $\pi_0$, for which $E_{\nu}(l(\varphi, \theta;Y,Z,\pi_0))$ is finite, that is
\begin{equation}\label{eq:valid_belief_update}
  \psi\{ l(\varphi, \theta;Y,Z,\pi_0), \pi_0 \}=\arg\min_\nu L(\nu;Y,Z,\pi_0).
\end{equation}
They first show that a valid belief update should minimise an overall loss $L$ of the form $L=E_\nu(l)+D(\nu,\pi_0)$, where the second term is a measure $D(\nu,\pi_0)$ of divergence between prior and $\nu$. For our purposes this actually defines what we mean by a ``valid" belief update. \cite{bissiriwalker12a} and \cite{Bissiri2016} show that if the loss $l$ is additive, and the belief update $\psi$ determined by \eqref{eq:valid_belief_update} is required to be order-coherent whatever the prior, parameter space, loss and data it is updating, if $L$ has a unique minimum and $D=D_g$ is a $g$-divergence (see Appendix~\ref{sec:preq_add_coherent_gives_post:proof}) then $D_g$ must be the KL-divergence and so a valid coherent belief update must minimise \eqref{eq:big_post_loss}. 

An optimal belief update minimising \eqref{eq:big_post_loss} exists when $E_{\pi_0}(\exp(-l(\varphi, \theta;Y,Z,\pi_0)))$ exists,
and if this holds then the optimal valid and coherent belief update is the proper Gibbs posterior in \eqref{eq:genbayes_expost_firsttime}. 
The result of \cite{Bissiri2016} justifies the belief update in \eqref{eq:genbayes_expost_firsttime} for an additive loss $l(\varphi, \theta;Y,Z)$. Theorem~\ref{thm:preq_add_coherent_gives_post} below extends this to prequentially additive losses.

\begin{Theorem}\label{thm:preq_add_coherent_gives_post}
  If a loss $l$ is prequentially additive with respect to the belief update given by the Gibbs posterior, 
  \begin{equation}\label{eq:psi-q-in-valid-thm}
  \psi^{(q)}(l(\varphi, \theta;Y,Z,\pi_0),\pi_0)\propto \exp(-l(\varphi, \theta;Y,Z,\pi_0)) \pi_0(\varphi, \theta)
  \end{equation}
  then $\psi^{(q)}$ is order-coherent. It further holds that $L(\nu; Y,Z,\pi_0)$ in \eqref{eq:big_post_loss}
  is a valid loss yielding an order-coherent belief update and $\psi^{(q)}$ itself is the optimal valid order-coherent belief update $\psi$ in \eqref{eq:valid_belief_update}.
\end{Theorem}

\begin{proof}
See Appendix~\ref{sec:preq_add_coherent_gives_post:proof}.
\end{proof}

Having a prior-dependent loss gives the discussion of valid belief updates a circular feeling. Prequential additivity  replaces additivity to determine (with coherence) a unique valid belief update. However, prequential additivity depends for its definition on some predefined rule $\psi^{(q)}$ for updating belief from $\pi_0$ to $\q_1$ and so on. The question remaining is whether prequential additivity and the coherence requirement are enough to impose a unique valid belief update, and whether that valid belief update coincides with the belief update $\psi^{(q)}$ which ensured the loss was prequentially additive.

We consider the Cut model as a first example of how this may be used to show the validity of a given belief update. 

\begin{corollary} \label{cor:cutisvalid}
  The Cut-model belief update defined in \eqref{eq:cutpost} is the optimal, valid and coherent belief update for the loss in \eqref{eq:cut_loss}.
\end{corollary}

\begin{proof}
It is sufficient by Theorem~\ref{thm:preq_add_coherent_gives_post} that the loss \eqref{eq:cut_loss} is prequentially additive with respect to the belief update \eqref{eq:cutpost}. This follows from Proposition~\ref{prop:cutmodeladditive}.
\end{proof}

\section{Semi-Modular Inference} \label{sec:smi}

Having established the Gibbs posterior as the valid belief update for the Cut-model loss,
we now point to some other related belief updates for prequentially additive losses. These are variants of $\eta$-SMI, a family of belief updates introduced in \cite{carmona20}. 
We define three families of candidate posterior distributions interpolating between the full-Bayes posterior \eqref{eq:bayespost} and the Cut-model posterior in \eqref{eq:cutpost}. The idea here, following \cite{carmona20}, is to provide modulated input to the $\varphi$ inference from the $(\varphi,\theta,Y)$-module. In the next section we motivate this step in a bit more detail.

\subsection{The Cut model and Bayesian Multiple Imputation}
\label{sec:bmi_cut}

\cite{Plummer2015} explains that the Cut-model approach to inference using \eqref{eq:cutpost} is Bayesian Multiple Imputation (BMI), in essence a two-stage process: at the imputation stage the posterior distribution $\pi(\varphi\mid Z)$ of $\varphi$ is imputed from the data $Z$ as if $\varphi$ were missing data; in the analysis stage the posterior distribution $\pi(\theta|Y,\varphi)$ of $\theta$ is conditioned on the imputed $\varphi$ so that uncertainty in $\varphi$ is fed through into the distribution of $\theta$.

Bayesian inference \eqref{eq:bayespost} can also be given formally as a two stage imputation/analysis procedure,
\begin{equation}\label{eq:bayes_as_BMI}
    \pi(\varphi,\theta \mid Y, Z) =\pi(\varphi \mid Y, Z) \pi(\theta\mid Y, \varphi ),
\end{equation}
using the posterior marginal
\[
\pi(\varphi|Y,Z)\propto \pi(\varphi)p(Z|\varphi)p(Y|\varphi)
\]
in the imputation stage, with $p(Y|\varphi)$ the marginal likelihood in \eqref{eq:marginal_Y_phi}. If we did carry out Bayesian inference in this way, we would use the same model, $p(Y|\varphi,\theta)$, in both $\pi(\varphi \mid Y, Z)$ (imputation) and $\pi(\theta\mid Y, \varphi )$ (analysis). This is an imputation scheme \cite{Meng1994} calls ``congenial'', where it is appropriate for the imputation and analysis to be carried out using the same model. In Cut-model inference the imputation and analysis use different models for $\varphi$, as $p(Y|\varphi,\theta)$ is not used in the imputation. This may help in what \cite{Meng1994} calls ``uncongenial'' problems.

One negative feature of the Cut model is that it may remove too much information from the imputation for $\varphi$. This will often increase the posterior variance of $\varphi$ and $\theta$. In the context of hypothesis tests based on classical multiple imputation of missing data, \cite{Knuiman1998} refer to this as ``dilution'' off the effect due to ``imputation noise''. We may be happy to accept some dilution, if the bias due to misspecification is substantial. However if the $(\varphi,\theta,Y)$ module is only weakly misspecified, we may see a large increase in variance for just a small bias.


\subsection{Semi-Modular inference and Tempered SMI}

The \emph{$\gamma$-SMI posterior} family of candidate posteriors simply tempers from the Cut (at $\gamma=0$) to full-Bayes (at $\gamma=1$) via
\begin{align}\label{eq:gamma_smi_post_tmp}
  \p^{(t)}_{\gamma}( \varphi,\theta \mid Y, Z ) & \propto  \p^{(c)}(\varphi,\theta \mid Y, Z)^{(1-\gamma)} \pi(\varphi,\theta \mid Y, Z)^\gamma \\
                                                &\propto\pi(\varphi,\theta \mid Y, Z)/p(Y|\varphi)^{1-\gamma},
\end{align}
using \eqref{eq:cut_from_post} for the last line. 
The loss function for which it is a Gibbs posterior is
\begin{equation} \label{eq:tmsmi_loss}
  \l^{(t)}( \varphi,\theta ; Y, Z ) = \l^{(b)}( \varphi,\theta ; Y, Z ) + (1-\gamma)\log p(Y \mid \varphi).
\end{equation}
The $p( Y \mid \varphi )$ term is the loss-function weighting that down-weights the influence of $Y$ on $\varphi$. We show in Section~\ref{sec:prop:gam-eta-del_are_preq_add:proof} that this loss is prequentially additive with respect to the belief update in \eqref{eq:gamma_smi_post_tmp}. It follows from Theorem~\ref{thm:preq_add_coherent_gives_post} that $\p^{(t)}_{\gamma}( \varphi,\theta \mid Y, Z )$ in \eqref{eq:gamma_smi_post_tmp} is the optimal, valid and coherent belief update for the loss in \eqref{eq:tmsmi_loss}.

The $\gamma$-SMI posterior in \eqref{eq:gamma_smi_post_tmp} is attractive as a formally straightforward family of candidate posteriors encompassing Cut models and Bayesian inference. However it is very awkward computationally and in fact we have no idea how to implement it in practice.
We now give two alternative interpolating sequences of candidate posterior distributions. The first is $\eta$-SMI, given in \cite{carmona20}. We begin by introducing an auxiliary parameter $\tilde\theta$, expanding the model parameters from $(\varphi,\theta)$ to $(\varphi,\tilde\theta,\theta)$.
The \emph{$\eta$-SMI posterior} is
\begin{equation} \label{eq:eta_smi_post_pow}
  \p^{(s)}_{\eta}(\varphi,\tilde\theta,\theta|Y, Z) = \p^{(s)}_{\eta}(\varphi,\tilde\theta|Y, Z) \pi(\theta|Y,\varphi)
\end{equation}
where $\p^{(s)}_{\eta}( \varphi , \tilde\theta \mid Y, Z )$ is a kind of power posterior
\begin{equation}\label{eq:pow_post}
  \p^{(s)}_{\eta}(\varphi, \tilde\theta \mid Y, Z ) = \frac{p(Z|\varphi) p( Y \mid \varphi, \tilde \theta )^\eta \;  \pi(\varphi,\tilde\theta)}{\int p(Z|\varphi) p( Y \mid \varphi, \tilde \theta )^\eta \;  \pi(\varphi,\tilde\theta) d\tilde\theta d\varphi}.
\end{equation}
Several authors (for example, \cite{Miller2018a}) observe that $p( Y \mid \varphi, \tilde \theta )^\eta$ is not a normalised probability density in $Y$. The power posterior is not simply a posterior distribution with an extra parameter $\eta$.

We are interested in the marginal belief update for $\theta$ and $\varphi$, which is
\begin{align} \label{eq:eta_smi_marg_post}
  \p^{(s)}_{\eta}(\varphi,\theta|Y, Z) & = \left[\int \p^{(s)}_{\eta}(\varphi,\tilde\theta|Y, Z)d\tilde\theta\right] \pi(\theta|Y,\varphi)  \nonumber                           \\
                                       & \propto \pi(\varphi) p(Z|\varphi)E_{\tilde\theta|\varphi}\left[ p(Y|\varphi,\tilde\theta)^\eta\right]  \pi(\theta|Y,\varphi).
\end{align}
The tempering or $\gamma$-SMI posterior in \eqref{eq:gamma_smi_post_tmp} can be written in a similar way
\begin{equation} \label{eq:gamma_smi_marg_post_phitheta}
  \p^{(t)}_{\gamma}(\varphi,\theta|Y, Z) \propto \pi(\varphi) p(Z|\varphi)E_{\tilde\theta|\varphi}\left[ p(Y|\varphi,\tilde\theta)\right]^\gamma  \pi(\theta|Y,\varphi),
\end{equation}
so the order of raising the power and marginalising is swapped.

The $\eta$-SMI posterior distribution $\p^{(s)}_{\eta}( \varphi,\theta | Y, Z )$ interpolates between Bayes, $\p^{(s)}_{\eta=1}(\varphi,\theta | Y, Z ) = p( \varphi,\theta |Y, Z )$ and Cut,
$\p^{(s)}_{\eta=0}(\varphi,\theta | Y, Z ) = \p^{(c)}( \varphi,\theta |Y, Z )$ (take $\eta=0,1$ in \eqref{eq:eta_smi_marg_post} and compare with Equations \eqref{eq:cutpost} and \eqref{eq:bayes_as_BMI}).

The loss function for which the $\eta$-SMI family of belief updates are Gibbs posteriors is
\begin{equation} \label{eq:smi_loss}
  \l^{(s)}( (\varphi,\tilde\theta,\theta) ; Y, Z ) =  \l^{(b)}( \varphi,\theta ; Y, Z ) -\eta \log p(Y \mid \varphi,\tilde\theta)  + \log p(Y \mid \varphi).
\end{equation}
We show in Section~\ref{sec:prop:gam-eta-del_are_preq_add:proof} that this loss is prequentially additive with respect to its Gibbs posterior, so that belief update is again the optimal, valid and coherent belief update.

\subsection{Kernel-Smoothing $\delta$-SMI}

The third interpolating sequence of candidate distributions we describe is constructed by taking a different relaxation of the likelihood. For $y,\y\in \R$ let $K_\delta(y,\y)$ be a normalised kernel. We focus on the cases $K_\delta(y,\y)=N(y-\y;0,\delta^2)$ and $K_\delta=(2\delta)^{-1}\I_{|y-\y|<\delta}$. For $y,\y\in \R^n$ we define 
\begin{equation}\label{eq:ks-kernel}
  K_\delta(y,\y)=\prod_{i=1}^n K_\delta(y_i,\y_i)
\end{equation}
and 
\begin{equation}\label{eq:ks-smooth-lkd}
  p_\delta( Y \mid \varphi, \tilde \theta )=\int p( \Y \mid \varphi, \tilde \theta ) K_\delta(Y,\Y)d\Y.
\end{equation}
Notice that if
\[
  p( Y \mid \varphi, \tilde \theta )=\prod_{i=1}^n p( Y_i \mid \varphi, \tilde \theta )
\]
then
\begin{equation}\label{eq:ks-smooth-lkd-prod}
  p_\delta( Y \mid \varphi, \tilde \theta )=\prod_{i=1}^n p_\delta( Y_i \mid \varphi, \tilde \theta )
\end{equation}
with
\begin{equation}\label{eq:ks-smooth-lkd-onesample}
    p_\delta( Y_i \mid \varphi, \tilde \theta )=\int p( \Y_i \mid \varphi, \tilde \theta )K_{\delta}(Y_i,\Y_i)d\Y_i.
\end{equation}
We define the \emph{$\delta$-SMI posterior} as
\begin{equation} \label{eq:kssmi_def}
  \p^{(k)}_{\delta}(\varphi,\tilde\theta,\theta|Y, Z) = \pi^{(k)}_{\delta}(\varphi,\tilde\theta|Y, Z) \pi(\theta|Y,\varphi)
\end{equation}
where $\pi^{(k)}_{\delta}( \varphi , \tilde\theta \mid Y, Z )$ is the kernel-smoothed posterior
\begin{equation}\label{eq:phi-side-ks-smi_def}
  \pi^{(k)}_{\delta}(\varphi , \tilde\theta \mid Y, Z ) =\frac{p(Z|\varphi) p_\delta( Y \mid \varphi, \tilde \theta ) \;  \pi(\varphi,\tilde\theta)}{\int p(Z|\varphi) p_\delta( Y \mid \varphi, \tilde \theta ) \;  \pi(\varphi,\tilde\theta) d\tilde\theta d\varphi}
\end{equation}
with $p_\delta( Y \mid \varphi, \tilde \theta )$ defined in \eqref{eq:ks-smooth-lkd}. 
We show in Section~\ref{sec:oc} that the $\delta$-SMI family of belief updates defined in \eqref{eq:kssmi_def} are valid for the loss,
\begin{equation} \label{eq:kssmi_loss}
  \l^{(k)}( (\varphi,\tilde\theta,\theta) ; Y, Z ) =  \l^{(b)}( \varphi,\theta ; Y, Z ) - \log p_\delta(Y \mid \varphi,\tilde\theta)  + \log p(Y \mid \varphi).
\end{equation}

\subsubsection{Interpretation of $\delta$-SMI as a generalised Cut model} In contrast to the likelihood relaxation $p( Y \mid \varphi, \tilde \theta )^\eta$ appearing in $\eta$-SMI, the likelihood $p_\delta( Y_i \mid \varphi, \tilde \theta )$ is a normalised density for $Y$, so the $\varphi,\tilde\theta$-posterior $\pi^{(k)}_{\delta}(\varphi , \tilde\theta \mid Y, Z )$ is a conditional probability (and so we write $\pi^{(k)}_{\delta}$ here). However, $\delta$-SMI as a whole is not simply Bayesian inference with some simple model elaboration.
The joint $\delta$-SMI posterior is in fact a cut model for an enlarged model with three modules.
The three data sets are $Y,Z$ and $Y'=Y$, the new copy of $Y$ present in the imputation stage for $\varphi$. The generative models for these three modules are $(\varphi,\tilde\theta,Y')\sim \pi(\varphi,\tilde\theta)p_\delta(Y'|\varphi,\tilde\theta)$, $(\varphi,Z)\sim \pi(\varphi)p(Z|\varphi)$ and $(\varphi,\theta,Y)\sim \pi(\varphi,\theta)p(Y|\varphi,\theta)$; the feedback from the final $\varphi,\theta,Y$ module into the $\varphi,\tilde\theta,Y'$ and $\varphi, Z$ modules has been cut. The posterior for the imputation stage is $\pi^{(k)}_{\delta}(\varphi,\tilde\theta|Y',Z)$ (with $Y'=Y$) and the posterior for the analysis stage is $\pi(\theta|Y,\varphi)$. 
This Cut-model interpretation does not hold for $\eta$-SMI, as $\p^{(s)}_\eta(\varphi,\tilde\theta|Y,Z)$ is not a posterior defined by Bayes rule, as $p(Y|\varphi,\tilde\theta)^\eta$ is not a normalised probability density.


\subsubsection{Comparison with $\eta$-SMI} We can display the relation between the marginal $\delta$-SMI posterior and the marginal $\eta$-SMI posterior. The marginal $\delta$-SMI posterior can be written
\begin{equation}\label{eq:kssmi_marginal_phitheta}
  \p^{(k)}_{\delta}(\varphi,\theta|Y, Z) \propto \pi(\varphi) p(Z|\varphi)E_{\tilde\theta|\varphi}\left[ p_\delta(Y|\varphi,\tilde\theta)\right]  \pi(\theta|Y,\varphi),
\end{equation}
so the $\delta$-SMI posterior looks like the $\eta$-SMI posterior in \eqref{eq:eta_smi_marg_post}, with prior expectation of the down-weighted likelihood $p_\delta(Y|\varphi,\tilde\theta)$ for the former and $p(Y|\varphi,\tilde\theta)^\eta$ in the later. 

\subsubsection{$\delta$-SMI interpolation of Bayes and Cut} Like $\eta$-SMI, the family of distributions indexed by $\delta$ interpolates between the Cut model and the Bayes posterior.

\begin{proposition} \label{prop:ks-smi-interpolates}
($\delta$-SMI interpolation)
If 
$\lim_{\delta\to 0}p_\delta(Y|\varphi,\theta)=p(Y|\varphi,\theta)$
and 
\begin{equation}\label{eq:ks-interpolation-DeltaInf}
\lim_{\delta\to \infty}\frac{p_\delta(Y|\varphi,\theta)}{p_\delta(Y|\varphi',\theta)}=1    
\end{equation}
for every $\varphi,\varphi'$
then the $\delta$-SMI posterior
$\p^{(k)}_{\delta}(\varphi,\theta | Y, Z )$ interpolates between Bayesian inference at $\delta=0$ and Cut-model inference as $\delta\rightarrow\infty$, that is
\[\lim_{\delta\rightarrow 0}\p^{(k)}_{\delta}(\varphi,\theta | Y, Z ) = \pi(\varphi,\theta |Y, Z )\]
and
\[
  \lim_{\delta\rightarrow\infty}\p^{(k)}_{\delta}(\varphi,\theta | Y, Z ) = \p^{(c)}( \varphi,\theta |Y, Z ).
\]
\end{proposition}
\begin{proof} 
Take the limits at $\delta=0$ and $\delta=1$ in \eqref{eq:kssmi_marginal_phitheta} using the stated behavior of $p_\delta(Y|\varphi,\theta)$ and compare against Equations \eqref{eq:cutpost} and \eqref{eq:bayes_as_BMI}. The likelihood $p_\delta(Y|\varphi,\theta)$ itself is improper at $\delta=\infty$. The condition \eqref{eq:ks-interpolation-DeltaInf} at $\delta\to\infty$ ensures that the limit of the posterior $\p_\delta$ exists and is equal to the cut model.
\end{proof}
 The conditions on the kernel smoothed likelihood in Proposition~\ref{prop:ks-smi-interpolates} restrict the choice of kernel $K_\delta$ in \eqref{eq:ks-kernel}. They are easily satisfied. For example, if the kernel $K_\delta$ is the top hat kernel and $Y_1,...,Y_n$ have a continuous density $p(Y_i|\tilde\theta,\varphi)$ then under the integral in \eqref{eq:ks-smooth-lkd}, we have $K_\delta(y,dy')\to \delta_{y}(dy')$ (the Dirac delta-function) as $\delta\to 0$ in the sense of a distribution, and $p_\delta(Y|\tilde\theta,\varphi)\to p(Y|\tilde\theta,\varphi)$. If $Y$ is discrete then for all sufficiently small $\delta$, the set $\{Y': |Y'_i-Y_i|\le \delta,\ i=1,...,n\}$ contains only $Y$ so $p_\delta(Y|\tilde\theta,\varphi)=p(Y|\tilde\theta,\varphi)$ for all sufficiently small $\delta$. Condition \eqref{eq:ks-interpolation-DeltaInf} also holds for the top-hat kernel. For example, for continuous real scalar data and $i=1,...,n$, $p_\delta(Y_i|\varphi,\tilde\theta)=(1-\epsilon_{\delta}(Y_i))/(2\delta)$ for some $\epsilon_{\delta}(Y_i)\to 0$ with $\delta\to\infty$ for any fixed data value $Y_i$ and so ratios tend to one.

\subsection{Targeting the $\delta$-SMI posterior} \label{sec:cut_mcmc}

\cite{carmona20} use the nested MCMC algorithm of \cite{Plummer2015} to target the $\eta$-SMI posterior $\p^{(s)}_{\eta}(\varphi,\theta|Y, Z)$. Here we show that similar methods can be setup to sample $\p^{(k)}_{\delta}(\varphi,\theta|Y, Z)$. \cite{Liu2020sacut} give an efficient approximation scheme which speeds up analysis within the same nested-MCMC framework.

We may not be able to compute the $\delta$-SMI likelihood
to $p_\delta(Y|\varphi, \tilde\theta)$. However we can treat the kernel $K_\delta$ as a probability density over ``missing'' data $\Y$, writing
\begin{equation}
  p_\delta(Y,\Y|\varphi, \tilde\theta)=K_\delta(Y,\Y)p(\Y|\varphi,\tilde\theta)
\end{equation}
so that the marginal obtained when we integrate over $\Y$ is $p_\delta(Y|\varphi, \tilde\theta)$ in \eqref{eq:ks-smooth-lkd}.
The extended posterior with auxiliary variables for the missing data is
\begin{equation}
  \p^{(k)}_{\delta}(\varphi,\tilde\theta,\theta,\Y|Y,Z)\propto \pi^{(k)}_{\delta}(\varphi, \tilde\theta,\Y|Y,Z)\pi(\theta|Y,\varphi)
\end{equation}
where
\[
  \pi^{(k)}_{\delta}(\varphi, \tilde\theta,\Y|Y,Z)\propto p(Z|\varphi)p_\delta(Y,\Y|\varphi, \theta)\pi(\varphi, \tilde\theta).
\]
The nested MCMC approach targets
\[
  \varphi, \tilde\theta,\Y\sim \pi^{(k)}_{\delta}(\varphi, \tilde\theta,\Y|Y,Z)
\]
using standard MCMC. Marginally then,
\[
  \varphi\sim \pi^{(k)}_{\delta}(\varphi|Y,Z).
\]
We take this simulated $\varphi$ and sample
\[
  \theta|\varphi\sim \pi(\theta|Y,\varphi)
\]
using standard MCMC. This gives a pair $(\varphi,\theta)$ distributed according to $\p^{(k)}_{\delta}(\varphi,\theta|Y,Z)$.
We do not use this Monte Carlo method below. In the main HPV-data example in Section~\ref{sec:hpv_analysis} below the likelihood $p_\delta(Y|\varphi, \tilde\theta)$ is given in terms of the CDF of a Poisson distribution and is readily evaluated.

The downside of this approach is that it suffers from ``double asymptotics". We run one MCMC chain generating samples from $\pi^{(k)}_{\delta}(\varphi|Y,Z)$. For each sample $\varphi$ output in this run we simulate a chain targeting $\pi(\theta|Y,\varphi)$ and take the last sampled $\theta$-value. This second chain must run to convergence. Whilst in our experience very high accuracy can be achieved in a modest runtime, of the order of ten times the runtime of the chain targeting the Bayes posterior $\pi(\varphi,\theta|Y,Z)$ for the same ESS \citep{carmona20}, this is clearly a weakness of this scheme. It may be preferable to analyse the $\delta$-SMI posterior using the variational framework of \cite{Yu2021variationalcut} and \cite{Carmona2021variationalSMI}.

\subsection{SMI and Bayesian Multiple Imputation}

Some of the forms of SMI listed above are variants of BMI in which we use information from the $Y$-module to inform the $\varphi$ imputation. This is the case for $\eta$-SMI and $\delta$-SMI. From a BMI perspective these SMI variants are simply trying to make the best possible imputation of $\varphi$ using the available information. The parameters $\eta$ and $\delta$ will be set to values that allow the right amount of information to flow back from $(\varphi,\theta,Y)$-module to influence the $\varphi$ imputation. The choice of these values is discussed in Section~\ref{sec:opt_eta}.
However, $\gamma$-SMI cannot be setup as BMI, at least in any computationally tractable way as it cannot be written as a suitable product of conditional probabilities.

\section{Properties of SMI}\label{sec:properties_of_SMI}

In this section we show that the different forms of SMI we have written down are all valid belief updates. We then give criteria and estimation procedures defining and computing an optimal $\delta$.

\subsection{Validity of new SMI variants} \label{sec:oc}

\cite{carmona20} show that $\eta$-SMI is order-coherent. The proof that its loss is prequentially additive is based on similar reasoning. We now extend these results to $\gamma$-SMI and $\delta$-SMI.

\begin{corollary} \label{cor:gam-eta-del_are_valid}
  The $\gamma$-SMI, $\eta$-SMI and $\delta$-SMI belief updates given respectively in \eqref{eq:gamma_smi_post_tmp}, \eqref{eq:eta_smi_post_pow} and \eqref{eq:kssmi_def} are the optimal, valid and coherent belief updates for their respective associated loss (see \eqref{eq:tmsmi_loss}, \eqref{eq:smi_loss} and \eqref{eq:kssmi_loss} respectively).
\end{corollary}
\begin{proof}
Since the losses are obtained from the corresponding Gibbs posteriors, it is sufficient by Theorem~\ref{thm:preq_add_coherent_gives_post} that these losses are prequentially additive with respect to their associated belief updates. This follows from Propostion~\ref{prop:gam-eta-del_are_preq_add} below.
\end{proof}

\begin{proposition} \label{prop:gam-eta-del_are_preq_add}
The loss functions for $\gamma$-SMI, $\eta$-SMI and $\delta$-SMI given respectively in \eqref{eq:tmsmi_loss}, \eqref{eq:smi_loss} and \eqref{eq:kssmi_loss} are prequentially additive with respect to the belief updates given respectively in \eqref{eq:gamma_smi_post_tmp}, \eqref{eq:eta_smi_post_pow} and \eqref{eq:kssmi_def}.
\end{proposition}

\begin{proof}
See Appendix~\ref{sec:prop:gam-eta-del_are_preq_add:proof}.
\end{proof}


\subsection{Asymptotic behaviour of $\delta$-SMI} \label{sec:asymptotics}

In Bayesian inference a family of densities $\mathbb{P}_\Omega=\{ p(\cdot|\varphi,\theta): (\varphi,\theta) \in\Omega \}$ with parameter space $\Omega$ is specified for unknown parameters $\theta,\varphi$ and belief about the true parameters $(\theta^*,\varphi^*)$ is updated by the observed data using Bayes' rule. If the model is well specified $p^*\in \mathbb{P}_\Omega$, then under regularity conditions, the posterior concentrates at the true parameter values as the number of observations increases. If the parametric model is misspecified $p^*\not\in \mathbb{P}_\Omega$ then, under regularity conditions, the posterior concentrates at the pseudo-true parameter values minimizing the Kullback-Leibler divergence between $p^*$ and $p(\cdot|\varphi,\theta)$ \citep{Berk1966}.
In these settings the Maximum Likelihood Estimator (MLE) is a natural estimator for the parameters minimising the Kullback-Leibler divergence \citep{Akaike1973}. The pseudo-truth is given by the limiting MLE taken on large data. The asymptotic behaviour of the Bayes posterior distribution for misspecified parametric models is considered in \citet{KleijnVaart2012}. A covariance matrix guaranteeing the correct asymptotic Frequentist coverage of the pseudo-true parameters was given by \citet{Muller2013}.

\cite{Pompe2021} give asymptotics for the Cut model. Because $\delta$-SMI is a kind of Cut-model inference (recall, the observation model $p_\delta(Y|\varphi,\tilde\theta)$ is normalised) that theory applies here. Denote by 
\begin{align}\label{eq:pseudo-true_general}
(\varphi^*_\delta,\tilde\theta^*_\delta)&=
\arg\max_{\varphi,\tilde\theta}E_{p^*(y,z)}(p(z|\varphi)p_\delta(y|\varphi,\tilde\theta))\\
\theta^*_\delta&=
\arg\max_{\theta}E_{p^*(y)}(p(y|\varphi^*_\delta,\theta))
\end{align}
the pseudo-true values of $\varphi,\tilde\theta$ and $\theta$ and let
\begin{align}\label{eq:MLE_general}
(\hat\varphi_\delta,\widehat{\tilde\theta}_\delta)&=
\arg\max_{\varphi,\tilde\theta} p(Z|\varphi)p_\delta(Y|\varphi,\tilde\theta))\\
\hat\theta_\delta&=
\arg\max_{\theta}p(Y|\hat\varphi_\delta,\theta)
\end{align}
be the separate MLE's in the imputation and analysis modules. \cite{Pompe2021} show that, under regularity conditions, and taking limits in $m$ and $n$ with $m/n=\alpha$, the cut-MLE's converge as 
\[
\sqrt{n}(\hat\varphi_\delta-\varphi^*_\delta,\widehat{\tilde\theta}_\delta-\tilde\theta^*_\delta,\hat\theta_\delta-\theta^*_\delta)\stackrel{D}{\longrightarrow} N(0,\Sigma_F),
\]
with $\Sigma_F$ a covariance defining asymptotic freqentist coverage of the pseudo-true values. \cite{Pompe2021} give this covariance in terms of the model elements. In contrast, if $\varphi,\tilde\theta,\theta\sim \p^{(k)}_\delta(\varphi,\tilde\theta,\theta|Y,Z)$
are samples from the $\delta$-SMI posterior then
\[
\sqrt{n}(\varphi-\hat\varphi_\delta,\tilde\theta-\widehat{\tilde\theta}_\delta,\theta-\hat\theta_\delta)\stackrel{D}{\longrightarrow} N(0,\Sigma_C),
\]
for some covariance $\Sigma_C$. \cite{Pompe2021} give $\Sigma_C$ in terms of the Cut-model elements. They show that $\Sigma_C\ne \Sigma_F$ in general, and so under the stated regularity conditions, the Cut-model posterior concentrates on the pseudo-true values, but does not have correct Frequentist coverage in the limit of large data. Since $\delta$-SMI is a kind of generalised Cut model (strictly a Cut at each $\delta$) the same observations apply.

\subsection{Choosing the influence parameter} \label{sec:opt_eta}

Having shown how to construct valid candidate posterior distributions for the Cut model and SMI, we select a candidate for downstream inference using an ``external" criterion. In this paper we select a candidate posterior by matching the posterior predictive distribution to the true generative distribution of the data. \cite{Wu2021calib} take a similar criterion when they select a power in the power posterior.

Following \cite{carmona20}, we consider \emph{out-of-sample predictive accuracy} of the model as our utility function for meta-parameter selection. Our criterion is the \textit{Expected Log Pointwise Predictive Density} (ELPD),
\begin{equation}\label{eq:ELPD}
  ELPD_{y,z}(Y,Z;\delta) = \int p^*(y, z) \log \p^{(k)}_{y,z,\delta}(y, z \mid Y, Z) dy dz,
\end{equation}
where $p^*$ is the distribution representing the true data-generating process and
\begin{equation}\label{eq:ks-smi-predictive}
  \p^{(k)}_{y,z,\delta}(y, z \mid Y, Z)=\int p(y, z \mid \varphi, \theta) \p^{(k)}_{\delta}(\varphi, \tilde\theta, \theta \mid Y,Z)\, d\varphi\,d\tilde\theta\,d\theta
\end{equation}
is a candidate posterior predictive distribution, indexed by $\delta$. We would like to set
\[
\delta^*=\arg\max_{\delta\ge 0} ELPD_{y,z}(Y,Z;\delta)
\]
and select the $\delta$-SMI posterior $p^{(k)}_{\delta^*}$ for analysis.
In general the ELPD must be estimated as $p^*$ is unknown.
In Section~\ref{sec:biased_data} (a simple synthetic example) we calculate the ELPD exactly. In Section~\ref{sec:another-example} we use LOOCV to estimate the ELPD (for the $Z$-data alone). In Section~\ref{sec:hpv_analysis} we use the WAIC to estimate the ELPD for the $Y$ and $Z$ data separately using the methods of \cite{Vehtari2016}.  

There is some freedom in the choice of utility function depending on the inference objective.
For example, in Section~\ref{sec:another-example} we use the ELPD for the $Z$ data alone as it prioritises $\varphi$-inference. One weakness of the ELPD is that we often value parameter estimation over predictive performance. It is not clear to us how to answer this issue in general. However, there are settings where one can take a utility which more directly targets parameter estimation. For example, if $\theta=(\theta_1,...,\theta_p)$ are model parameters which enter as a priori exchangeable auxiliary variables naturally interpreted as missing data, and the data $Y,Z$ comes with actual observations of a subset $(\theta_1,...,\theta_d),\ 1\le d<p$ of the $\varphi$-values, then we may choose $\delta$ using LOOCV, treating the observed $\theta$-values as the held-out data. See \cite{carmona22} for an example where this approach is used.

\section{Examples} \label{sec:data_analyses}

Here we present three reproducible examples. R code \citep{Rcore19} reproducing all results below is given in \verb|https://github.com/gknicholls/delta-SMI-repository|.


\subsection{Simulation study: Biased data} \label{sec:biased_data}

This is a simple synthetic example taken from \cite{Liu2009} in which the source of the ``misspecification'' is a poorly chosen prior. Since there is no misspecification in the observation models the interpolating models $\p^{(k)}_\delta,\ \delta\ge 0$ (including Cut and Bayes) concentrate, in the limit $n\rightarrow\infty$ with $m/n$ constant, on the true parameter values $(\varphi^*,\theta^*)$. The KL-divergence between $p^*$ and $\p^{(k)}_{y,z,\delta}$ tends to zero and the ELPD converges to a constant $\int p^*\log(p^*)dydz$ independent of $\delta$. 

Suppose we have two datasets informing an unknown parameter $\varphi$. The first is a ``reliable'' small sample $Z=(Z_1,\ldots,Z_n),\ Z_j\sim N( \varphi, \sigma_z^2 )$, iid for $j=1,...,m$ distribution, with $\sigma_z$ known; the second is a larger sample $Y=(Y_1,\ldots,Y_n), Y_i\sim N( \varphi + \theta , \sigma_y^2 )$ iid for $i=1,...,n$, with $\sigma_y$ known. The ``bias'' $\theta$ is unknown.

This model was given in \cite{Jacob2017b} as an example where Cut model approaches improve on Bayesian inference and analysed in \cite{carmona20} as an example of $\eta$-SMI. Here we repeat this analysis for our new SMI variants. In this normal setup the three interpolations $\eta$-SMI, $\delta$-SMI and $\gamma$-SMI are all identical. We may take a fixed value of $\delta$ and recover the $\eta$-SMI and $\gamma$-SMI distributions by setting $\eta=\sigma_y^2/(\sigma_y^2+\delta^2)$ and $\gamma=\frac{\sigma_y^2}{\delta^2+\sigma_y^2}+1$. 

We choose true parameter values in such a way that each dataset offers apparent advantages to estimate $\varphi$. One dataset is unbiased but has a small sample size, $m=25$, whereas the second has an unknown bias but more samples, $n=50$, and smaller variance. Suppose the true generative parameters are $\varphi^*=0$, $\theta^*=1$, and we know $\sigma_z=2$ and $\sigma_y=1$. We assign a constant prior for $\varphi$, while $\theta$ is subjectively assessed to have a $N(0, \sigma_{\theta}^2)$ prior. We are over-optimistic about the size of the bias and set $\sigma_\theta=0.33$. These choices are all the same as previous authors except that those authors took $\sigma_\theta=0.5$. Our choice is a little more ``extreme''. We do this simply to get an example where effects are a bit more visible.

We calculate the $\delta$-SMI posterior for a range of $\delta\in[0,\infty]$. Picking up from the marginal \eqref{eq:kssmi_marginal_phitheta} of interest,
\[
  \p^{(k)}_{\delta}(\varphi,\theta|Y,Z)= \pi^{(k)}_{\delta}(\varphi|Y,Z) \pi(\theta|Y,\varphi)
\]
where the posterior for $\theta$ given $\varphi$ is
\[
  \pi(\theta|Y,\varphi)=N(\theta;\mu_{\theta|Y,\varphi},\sigma^2_{\theta|Y,\varphi}),
\]
with
\[
  \mu_{\theta|Y,\varphi}=\rho(\bar Y-\varphi),\quad
  \sigma_{\theta|Y,\varphi}^2=(1-\rho)\sigma_\theta^2,\quad \rho= \frac{\sigma_\theta^2}{\sigma_\theta^2+\sigma_y^2/n}
\]
The marginal posterior for $\varphi$  is
\begin{align*}
  \pi^{(k)}_{\delta}(\varphi|Y,Z) & =\int \pi^{(k)}_{\delta}(\varphi,\tilde\theta|Y,Z)d\tilde\theta \\
                                 & =N(\varphi;\mu_\delta,\sigma^2_\delta),
\end{align*}
with
\begin{equation}\label{eq:biased_normal_phi_post_pars}
  \mu_\delta=\lambda \bar Z+(1-\lambda)\bar Y,\quad
  \sigma^2_\delta=\lambda \sigma_z^2/m,\quad
  \lambda= \frac{m/\sigma_z^2}{m/\sigma_z^2+n/(\sigma_y^2+\delta^2+n\sigma_\theta^2)}
\end{equation}
With these expressions we have $\p^{(k)}_{\delta}(\varphi,\theta|Y,Z)$ as a product of normal densities. This may be sampled by simulating
\begin{align*}
  \varphi        & \sim N(\varphi;\mu_\delta,\sigma^2_\delta)                         \\
  \theta|\varphi & \sim N(\theta;\mu_{\theta|Y,\varphi},\sigma^2_{\theta|Y,\varphi}).
\end{align*}
at any desired value of $\delta$. We get the Bayes and Cut posteriors by taking the respective limits $\delta\rightarrow 0$ and $\delta\rightarrow \infty$.
A scatter plot of the $\p^{(k)}_{\delta}$ posterior at three values of $\delta=0,\delta^*,\infty$ is given in Fig.~\ref{fig:NormEx-Scatter}.
\begin{figure}[!ht]
  \begin{center}
    \includegraphics[width=0.6\textwidth]{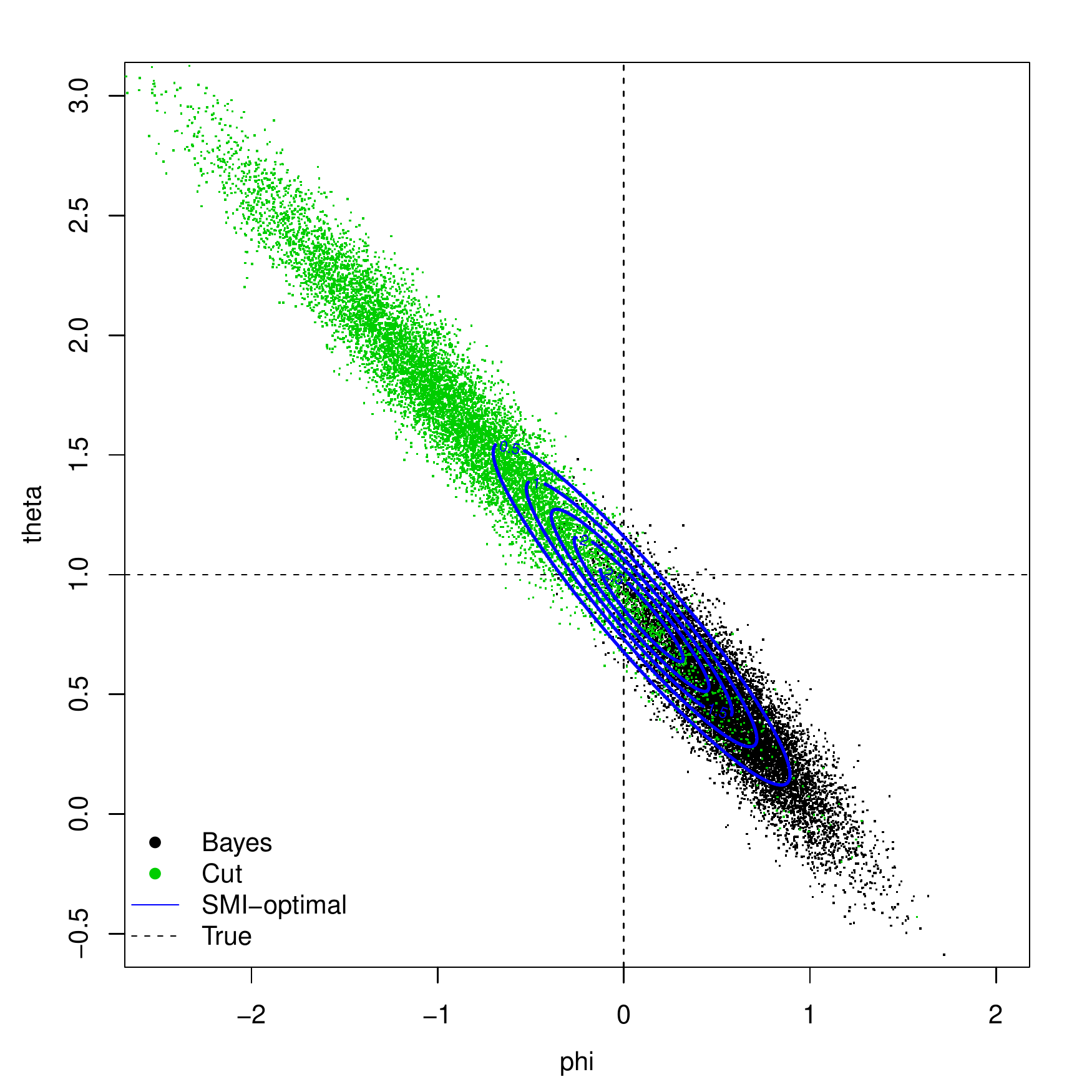}
  \end{center}
  \caption[Biased data]{Candidate posterior distributions for normal biased data example. The points are samples from $\p^{(k)}_{\delta}(\varphi,\theta|Y,Z)$ for three values of $\delta=0,\delta^*$ and $\delta=\infty$ yielding the Bayes, Cut and optimal $\delta$-SMI posterior. The dotted lines show true parameter values. }
  \label{fig:NormEx-Scatter}
\end{figure}
The $\delta$-SMI posterior covers the truth. For ease of visualisation the random number seed was chosen (six attempts) so that the Cut and Bayes posteriors were reasonably well separated, but in other respects this is typical. The $\delta$-SMI posterior does relatively well for recovering the true parameters, though it is chosen by targeting the ELPD. This is not expected, or even desirable, in a misspecified model. However, in this example the observation models are both exactly correct, and the misspecification is in the prior.

For further visualisation we plot in Fig.~\ref{fig:SMI_biased_data} the marginal $\delta$-SMI posteriors for $\varphi$ (top) and $\theta$ (bottom) at $\delta=0$ (Bayes) and $\delta=\infty$ (Cut) together with the selected $\delta$-SMI at $\delta^*$, the choice maximising the ELPD. In this example where only the $\theta$-prior is misspecified,
Bayes has little overlap on the truth. Cut has reasonable overlap but larger variance, as the $Y$ data do not inform $\varphi$. The $\delta$-SMI posterior selected using the ELPD has lower variance than the Cut {\it and} better location.
The data are synthetic, so we estimate the Posterior Mean Squared Errors (PMSE) $E_{\p^{(k)}_{\delta^*}}[(\varphi-\varphi^*)^2\mid Y,Z]$ and $E_{\p^{(k)}_{\delta^*}}[(\theta-\theta^*)^2\mid Y,Z]$ measuring the dispersion of the selected $\delta$-SMI posterior around the truth, and calculate the posterior predictive distribution for new data and the exact ELPD in Appendix~\ref{app:biased_data} using \eqref{eq:ks-smi-predictive} and \eqref{eq:ELPD}.
This simple example would be quite challenging for Monte-Carlo estimation of $ELPD_{y,z}(Y,Z;\delta)$. Referring to Fig.~\ref{fig:SMI_biased_data} the variation in the ELPD (bottom right panel) with $\delta$ is small, so its maximum is hard to locate accurately.

In the right column of Fig.~\ref{fig:SMI_biased_data} we display these metrics for $\delta\in[0,\infty]$. The values of PMSE and $ELPD_{y,z}(Y,Z;\delta)$ for Bayes and Cut correspond respectively to the values taken by the functions plotted at the left and right edges of the graphs. We see their PMSE's are larger (as we would expect from the marginal posterior densities) and their ELPD-values are lower than those of the $\delta$-SMI posterior.
\begin{figure}[!ht]
  \begin{center}
    \includegraphics[width=0.48\textwidth]{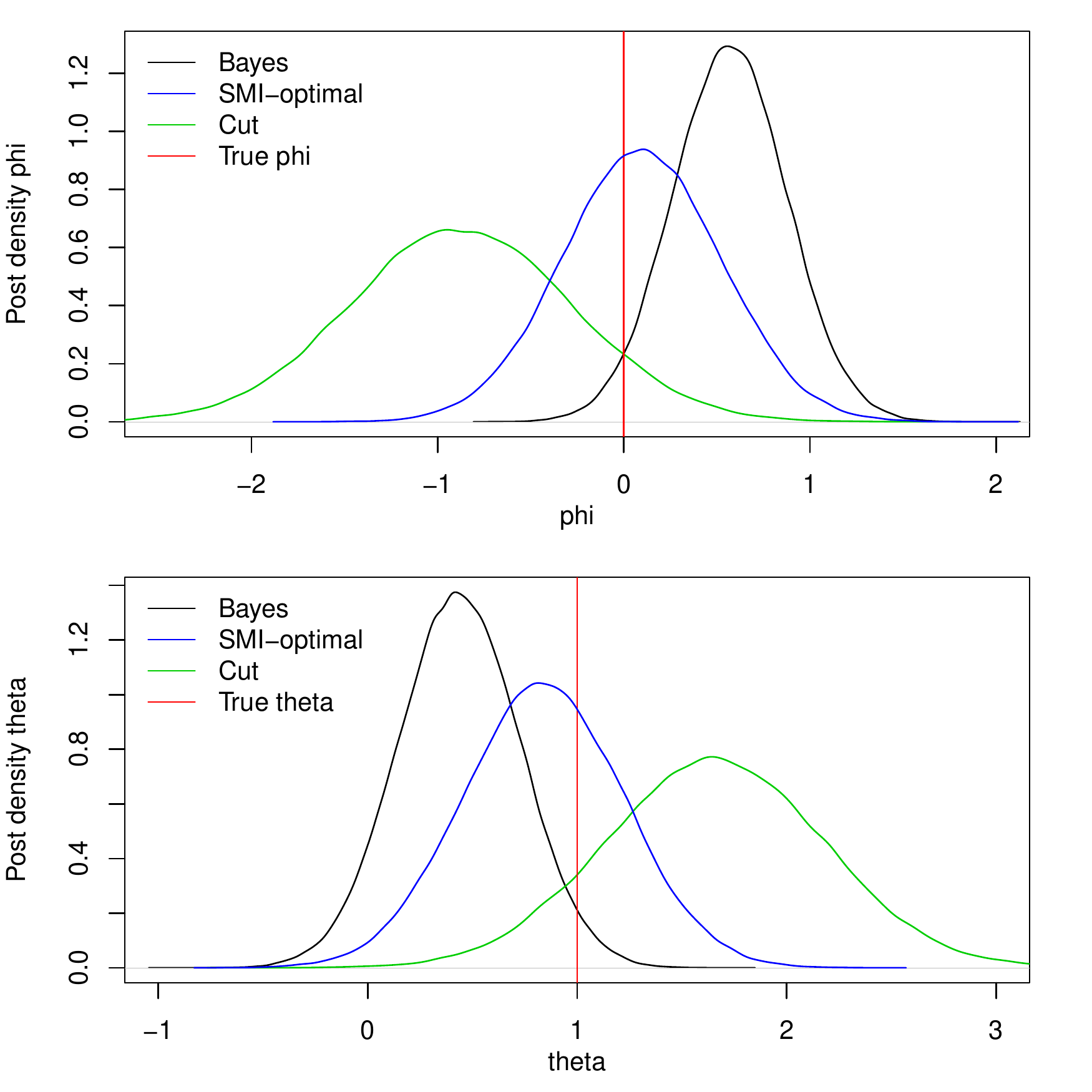}\includegraphics[width=0.48\textwidth]{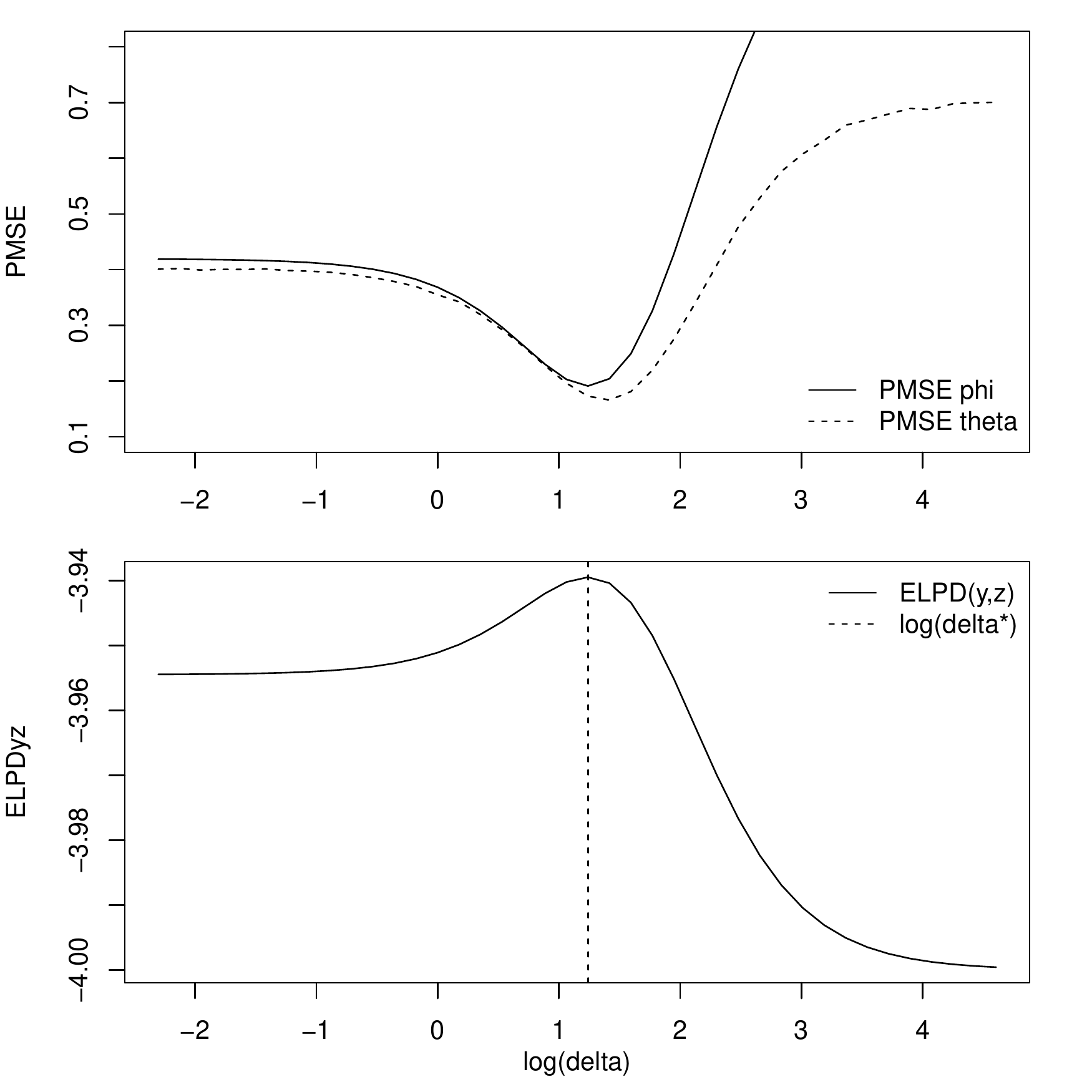}
  \end{center}
  \caption[Biased data]{Model assessment for normal biased data example. Left column: $\delta$-SMI posteriors for $\varphi$ (top) and $\theta$ (bottom) showing the Cut (green) and Bayes (black) and selected $\delta^*$-SMI posterior (blue) with the true parameter values indicated by a vertical line (red). Top-right: PMSE's for $\varphi$ (solid) and $\theta$ (dashed) as a function of the meta-parameter $\log(\delta)$. Bottom-Right: the $ELPD_{y,z}(Y,Z;\delta)$ as a function of $\log(\delta)$. The selected meta-parameter value $\delta^*$ is indicated by the vertical dashed line.}
  \label{fig:SMI_biased_data}
\end{figure}

The estimated value of $\delta^*\simeq 3.5$ in $\delta$-SMI corresponds to $\eta^*\simeq 0.08$ in $\eta$-SMI. The scale of the ``noise'' added to the $Y$-values looks relatively large compared to their variance $\sigma^2_y=1$. This tells us that the $Y$-module is misspecified. However, referring to \eqref{eq:biased_normal_phi_post_pars} we see ${\delta^*}^2\simeq 12$ is large relative to $\sigma_y^2+n\sigma_\theta^2=6.4$, so $\delta$-SMI is actually removing information from the $\theta$-prior from the $\varphi$-imputation. 



\subsection{Misspecified Regression model}\label{sec:another-example}

This simple synthetic example illustrates the behavior of the method when the observation model in the $Y$-module is misspecified. The setup is otherwise similar to the biased-data example. We have a well specified $Z$-module with a small data set. Interest focuses on estimation of $\varphi$. We have a second larger data set (the $Y$module). Standard Bayesian analysis has given us reason to believe the $Y$-module is misspecified so we cannot estimate $\theta$. However, we will use some information from the $Y$-module in order to reduce the variance of our $\varphi$-estimation. We use $\delta$-SMI to control the bias coming from the misspecified $Y$-module.

The model is a regression. Covariates $X_i\sim F_X, i=1,...,n$ and their sampling distribution $F_X$ are known exactly. The fitted models are
\begin{equation}\label{eq:fitted-simple-reg-model}
\begin{aligned}
    Y_i&\sim N(\varphi+\theta X_i,\sigma_y^2),\ i=1,...,n,\\
    Z_j&\sim N(\varphi,\sigma_z^2),\ j=1,...,m.
\end{aligned}
\end{equation}
The true parameter values are $\varphi^*,\theta^*$. The true observation model for $Z$ is the same as the fitted model,
\[
Z_j\sim N(\varphi^*,\sigma_z^2),\ j=1,...,m.
\]
The true model for $Y=(Y_1,...,Y_n)$ is
\begin{equation}\label{eq:true-simple-reg-model}
Y_i\sim N(\varphi^*+\theta^* X_i^k,\sigma_y^2),\ i=1,...,n.
\end{equation}
with $k>0$ a parameter we vary to illustrate different levels of misspecification.
The $\varphi$ and $\theta$ priors are both flat improper priors. Parameter settings are given in Appendix~\ref{app:regression_example}.

The resulting $\delta$-SMI distributions (once integrated over $\tilde\theta$) are
\[
\p_\delta(\varphi|Y,Z,X)=N(\varphi; \tilde\mu_\varphi, \tilde\sigma_\varphi^2),
\]
with 
\[    \tilde\mu_\varphi=\frac{\rho \widebar{z} +\widebar{y}\left(1-\frac{\widebar{x} \widebar{xy}}{\widebar{x^2}\widebar{y}}\right)}{\rho+1-\frac{\widebar{x}^2}{\widebar{x^2}}}\qquad
    \tilde\sigma_\varphi^2=\frac{\rho\sigma_z^2/m}{\rho+1-\frac{\widebar{x}^2}{\widebar{x^2}}}
    \]
where
\[
\rho=\frac{\sigma_y^2+\delta^2}{\sigma_z^2}\times \frac{m}{n}
\]
and
\[
\pi(\theta|Y,X)=N(\theta; \tilde\mu_{\theta|\varphi}, \tilde\sigma_{\theta|\varphi}^2),
\]
with 
\[    \tilde\mu_{\theta|\varphi}=\frac{\widebar{xy}-\varphi\widebar{x}}{\widebar{x^2}},\qquad
    \tilde\sigma_{\theta|\varphi}^2=\frac{\sigma_y^2}{n\widebar{x^2}}.
\]
The joint $\delta$-SMI posterior is then
\begin{equation}\label{eq:regression_example_SMI_post}
\p_\delta(\varphi,\theta|Y,Z,X)\, =\, N(\varphi; \tilde\mu_\varphi, \tilde\sigma_\varphi^2)\ \times\ N(\theta; \tilde\mu_{\theta|\varphi}, \tilde\sigma_{\theta|\varphi}^2)
\end{equation}

The MLE's \eqref{eq:MLE_general} obtained by maximising the likelihoods on each side of the cut coincide with the posterior means above, $\hat\varphi_\delta=\tilde\mu_\varphi$ and $\hat\theta_\delta=\tilde\mu_{\theta|\hat\varphi_\delta}$ (the MLE's $\widehat{\tilde\theta}_\delta=\hat\theta_\delta$ are equal).
These converge to the pseudo-true values defined in \eqref{eq:pseudo-true_general} and given here by
\begin{equation}\label{A23:1}
\theta^*_\delta = \frac{\theta^* M_{X^{k+1}}+M_X\varphi^*-M_X \varphi^*_\delta}{M_{X^2}}
\end{equation}
and 
\begin{equation}\label{A23:2}
\varphi^*_\delta= \varphi^*+\theta^* \frac{  M_{X^2}M_{X^k}-M_X M_{X^{k+1}}}{Var(X)+\alpha M_{X^2}(\sigma_y^2+\delta^2)/\sigma_z^2},
\end{equation}
where $M_{X^r}=E(X^r),\ r=1,2,...$  where $X\sim F_X$ is the scalar covariate (an abuse of notation).
Since the posterior means and MLE's coincide, the $\delta$-SMI posterior in \eqref{eq:regression_example_SMI_post} converges as $n\to \infty$ with $\alpha=m/n$ fixed to concentrate on the pseudo-true values. 
It is clear from the pseudo-true expressions that $\sigma_y^2+\delta^2$ balances $\alpha=m/n$, so larger $\delta$ gives smaller effective $Y$-sample size $n$. 
If $\delta^2=c/\alpha-\sigma_y^2$ for fixed $c>0$, then the posterior concentrates on the same pair of $(\varphi,\theta)$-values as $\alpha$ varies. 
As $\delta\to \infty$ (cut model), $\varphi_\delta\to \varphi^*$ approaches the true value as the $Z$-model is not misspecified, and $\theta_\delta\to \theta^*M_{X^{k+1}}/M_{X^2}$.  If $k=1$ there is no model-misspecification, and the pseudo-truth approach to the true values regardless to $\delta$ and $\alpha$ values. 

In this setting, if we are interested in estimating $\varphi$ then we use the ELPD for $z$ alone to define the optimal $\delta^*$-value as it favors a posterior $\p^{(k)}_\delta$ concentrated on $\varphi^*$. It is given by
\begin{equation}\label{eq:ELPD_z}
    ELPD_z(Y,Z;\eta)=\int p^*(z)\log(\p^{(k)}_\delta(z|Y,Z))dz.
\end{equation}
We can calculate the exact $ELPD_z$ in this example. However, in order to show how well the method works in practice, we instead estimate $ELPD_z$ in \eqref{eq:ELPD_z} using the LOOCV estimator
\begin{equation}\label{eq:ELPD_z_LOOCV_est}
    \widehat{ELPD}_z(Y,Z;\delta)= \frac{1}{m}\sum_{j=1}^m\log(\p^{(k)}_\delta(Z_j|Y,Z_{-j})).
\end{equation}
We set $\delta^*=\arg\max_{\delta\ge 0} \widehat{ELPD}_z(Y,Z;\delta)$. We then estimate
the posterior mean square errors (PMSE) $PMSE_\varphi=E_{\p^{(k)}_{\delta^*}}[(\varphi-\varphi^*)^2\mid Y,Z]$
using $S$ posterior samples 
$\varphi^{(s)}\sim \p^{(k)}_{\delta^*}(\varphi|Y,Z),\ s=1,...,S$ so that
\[
\widehat{PMSE}_\varphi=\frac{1}{S}\sum_{s=1}^S(\varphi^{(s)}-\varphi^*)^2.
\]
In Figure~\ref{fig:MSE_regression_example} (top) we show how the posterior mean squared error varies as we increase the level of misspecification by varying $k$ from $k=1$ (no misfit) up to $k=2$ (linear fit to quadratic). Each box shows the scatter of 100 $\widehat{PMSE}_\varphi$-values estimated using 100 independent replicate data sets and associated $\delta^*$-values. At large $k\simeq 2$ the Cut model (green) gives a lower PMSE than Bayes (red). When $k\simeq 1$ the Bayes posterior is more concentrated on the true parameter. The LOOCV-selected $\delta$-SMI posterior $\p^{(k)}_{\delta^*}$ tracks the ``best" of these two as $k$ varies. One question is whether allowing $\delta$ to take values other than $0$ or $\infty$ is actually adding anything. Our EPLD-utility targets prediction so of course $\delta$-SMI does well on this criteria whilst the PMSE-gains are slight. Figure~\ref{fig:MSE_regression_example} (bottom) compares the exact $ELPD_z$ of the selected $\delta$-SMI posterior with Bayes and Cut and shows the clear benefit of $\delta$-SMI. This amounts to a test of the quality of the LOOCV estimation of $ELPD_z$ in \eqref{eq:ELPD_z_LOOCV_est}. 

There may be some advantage in using $\delta^*$ as a summative measure of misspecification. If it is very small, or very large, we use Bayes or Cut respectively. However, for intermediate values of $k$ in Figure~\ref{fig:MSE_regression_example} we see that $\delta$-SMI does slightly better than Bayes or Cut. For this range of $k$, the Bayes and Cut distributions are far apart, but the misspecification is not so bad that we gain by simply cutting feedback altogether. \cite{carmona20} give an example for $\eta$-SMI in which more dramatic gains are seen from using intermediate values.
\begin{figure}
    \centering
    \includegraphics[width=4.5in]{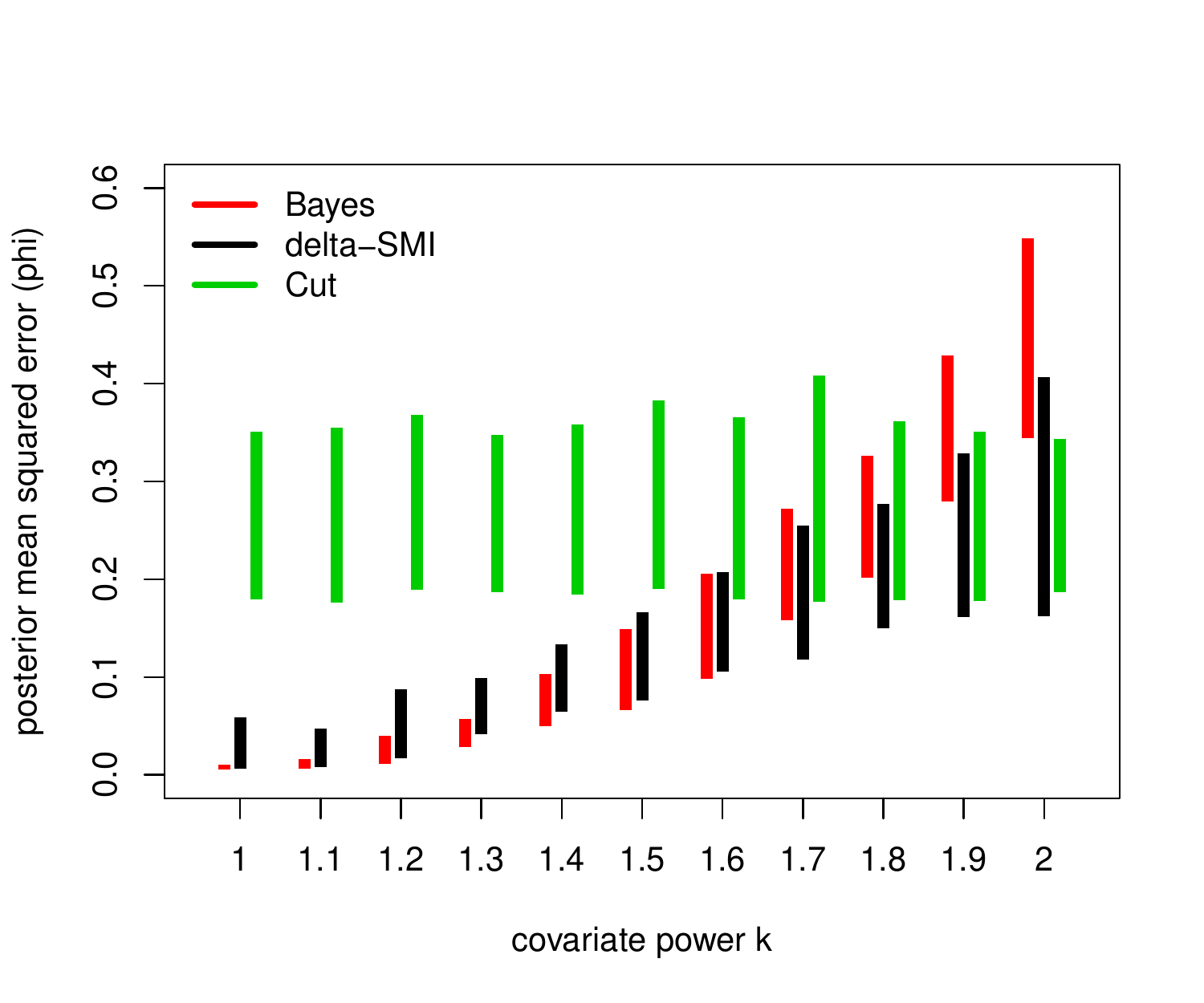}
    \includegraphics[width=4.5in]{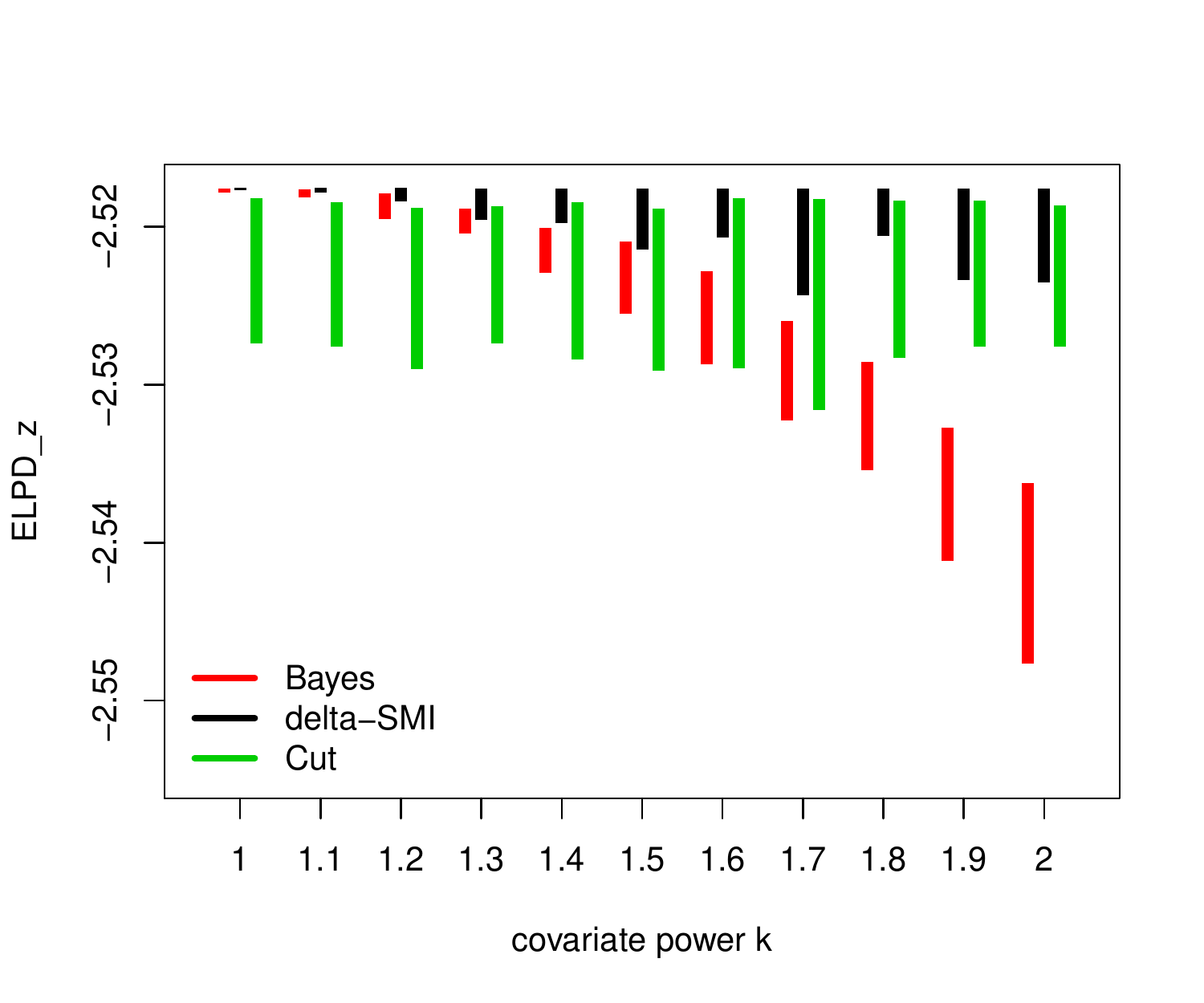}
    \caption{(top) Posterior mean squared error $\widehat{PMSE}_\varphi$ as a function of $k$, the covariate power in the true observation model mean $\varphi^*+\theta^* X^k$. 
    (bottom) Exact $ELPD_z$ of $\delta$-SMI at $\delta^*$ compared to Bayes and Cut.
    These are boxplots with transparent whiskers and outliers.  Each box summarises 100  $\widehat{PMSE}_\varphi$-values computed from 100 independent data sets. Cut model (green), Bayes (red), SMI (black). 
    }
    \label{fig:MSE_regression_example}
\end{figure}

\subsection{Epidemiological data} \label{sec:hpv_analysis}
In our final example, we apply SMI to an epidemiological dataset introduced by \cite{Maucort-Boulch2008}, studying the correlation between human papilloma virus (HPV) prevalence and cervical cancer incidence, revisited by several authors including \cite{Plummer2015} and \cite{Jacob2017b} in the context of Cut models and \cite{carmona20} for $\eta$-SMI.

The model has two modules: in each population $i=1,...,n,\ n=13$, a Poisson response for the number of cancer cases $Y_i$ in $T_i$ women-years of followup, and a Binomial model for the number $Z_i$ of women infected with HPV in a sample of size $N_i$ from the $i$'th population. For $i=1,...,n$,
\begin{align*}\label{eq:HPV_model}
  Y_i &\sim \text{Poisson}( \mu_i ) \\
  \mu_i &= T_i \exp( \theta_1+\theta_2 \varphi_i ) \nonumber \\
  Z_i &\sim \text{Binomial}(N_i, \varphi_i ). \nonumber
\end{align*}
There are reasons to expect the Poisson module to be misspecified \citep{Plummer2015}.
The relaxation of the Poisson likelihood under $\delta$-SMI is defined by 
\begin{equation}\label{eq:deltaPoisLik}
p_\delta( Y_i \mid \varphi, \tilde \theta ) = \sum_{\Y_i=0}^\infty p( \Y_i \mid \varphi, \tilde \theta )K_{\delta}(Y_i,\Y_i).
\end{equation}
The kernel $K_{\delta}$ is a discrete uniform distribution over the $\delta$-neighborhood 
$\b(Y_i,\delta)$ of $Y_i$,
\begin{equation}
	\b(Y_i,\delta) = 
   \crvBr{\y\in\mathbb{Z}_0^+: |\y-Y_i|\le \delta }
   \label{eq:unscaledDeltaNbr}
\end{equation}
so that
\begin{equation}
K_\delta(Y_i,\Y_i) = \frac{\mathbb{I}_{\Y_i\in \b(Y_i,\delta)}}{|\b(Y_i,\delta)|},\ i=1,...,n. \nonumber
\end{equation}
Let $\b(Y_i, \delta)_+=\max(\b(Y_i, \delta)$ (equal $\lfloor Y_i+\delta\rfloor$ here) and $\b(Y_i, \delta)_-=\min(\b(Y_i, \delta)$ (equal $\max(0,\lceil Y_i-\delta\rceil)$ here). Equation \eqref{eq:deltaPoisLik} becomes, for $i=1,...,n$,
\begin{equation}
p_\delta( Y_i \mid \varphi, \tilde \theta ) = F( \b(Y_i, \delta)_+ \mid \varphi_i, \tilde \theta ) - F( \b(Y_i, \delta)_- -1 \mid \varphi_i, \tilde \theta ),
\nonumber
\end{equation}
where $F(\cdot\mid \varphi_i, \tilde \theta)$ is the Poisson CDF with mean $\mu_i$. Notice that when $\delta<1$ the set $\b(Y_i,\delta)=\{Y_i\}$ contains only the observed data so $p_\delta( Y_i \mid \varphi, \tilde \theta )=p( Y_i \mid \varphi, \tilde \theta )$ for that range of $\delta$-values, as observed below Proposition~\ref{prop:ks-smi-interpolates}.

Following \cite{carmona20}, we use the $ELPD_y$ of the Poisson data (where $ELPD_y$ is defined in a similar way to $ELPD_z$ in \eqref{eq:ELPD_z}) as
estimated by WAIC \cite{Vehtari2016} to select the $\delta$-SMI distribution $\p^{(k)}_{\delta^*}$ with posterior predictive distribution most closely matching the true generative model, and compare against $\eta$-SMI, with $\eta^*$ chosen in the same way. Nested MCMC targeting the the $\delta$-SMI posterior was implemented using STAN \citep{stan2017}.

Fig.~\ref{fig:hpvCompareThetaDistr} 
\begin{figure}
\centering
\includegraphics[width=12.5cm]{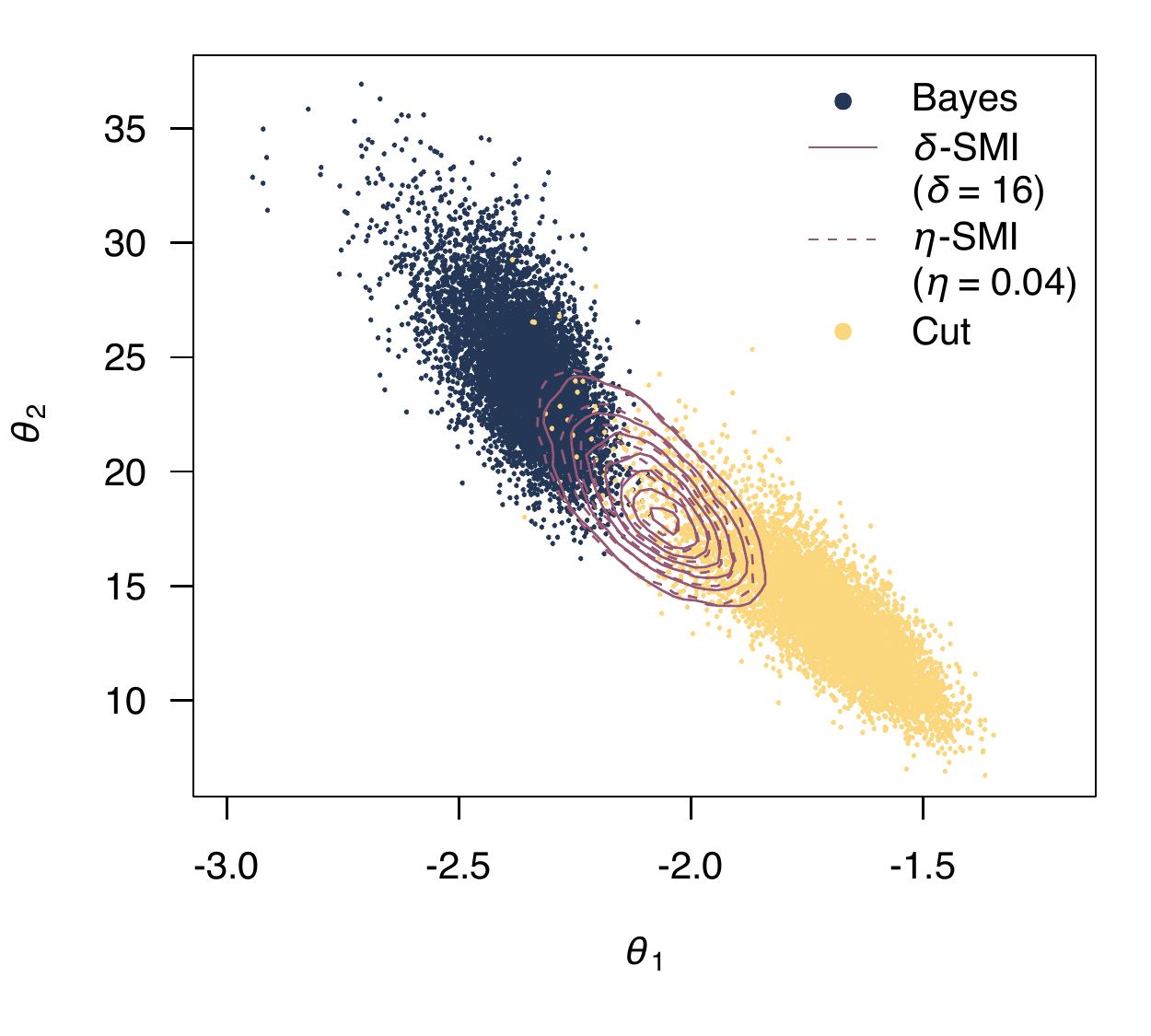} 
\caption{Estimates of the joint distribution of $\theta_1$ and $\theta_2$ in $\delta$- and $\eta$-SMI. 
The Bayes posterior is represented by navy blue dots, the Cut model by yellow dots, $\delta$-SMI by solid purple contour lines, and $\eta$-SMI by dashed purple contour lines. The $\delta$ and $\eta$ values were selected for illustration only. Their SMI distributions are visually distinct from the Bayesian and the Cut-model posteriors and they give comparable values of ELPD (estimated by WAIC). }
\label{fig:hpvCompareThetaDistr}
\end{figure}
presents the joint distribution of $\theta_1$ and $\theta_2$ estimated from the full Bayes model, Cut model, $\delta$-SMI and $\eta$-SMI. The Bayes (navy blue) and Cut-model (yellow) posteriors are well separated. 
The two candidate $\delta$-SMI and $\eta$-SMI distributions (purple and dashed-purple contours) in this figure are not those at $\delta^*$ and $\eta^*$ respectively. Instead, we choose a ``central" $\delta$-value and then choose a corresponding $\eta$ with a comparable $ELPD_y$-value. This is done to show how similar the $\delta$- and $\eta$-SMI posteriors for $\theta_1$ and $\theta_2$ are across the range of candidate posteriors, when we match them by their ELPD values. The $\delta$- and $\eta$-SMI posteriors are of course identical at Bayes and Cut and this shows how similar they are over the range.

Figure~\ref{fig:hpvWaicEtaVsDelta} 
\begin{figure}
\centering
\includegraphics[width=12.5cm]{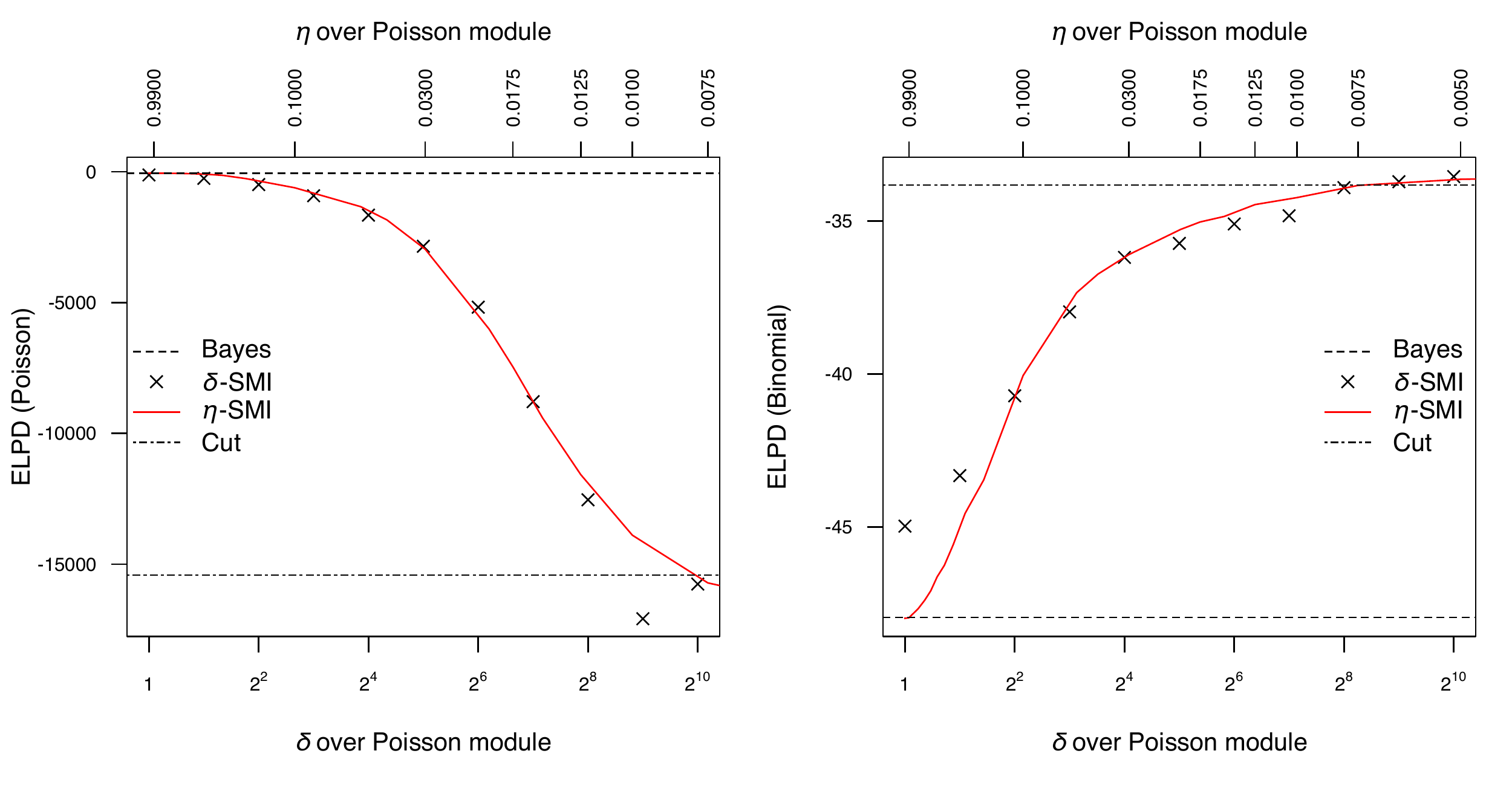} 
\caption{Estimated ELPD values (using WAIC) as predictive criteria for selection of $\eta \in \sqrBr{0, 1}$ and $\delta \in \crvBr{0, 1, ...}$ for the HPV model.
The $\eta$ values are transformed by non-linear but monotone regression to match the $\delta$-SMI ELPD-values.
The left panel shows $ELPD_y$ and is computed on the Poisson $Y$-data alone. The right panel shows $ELPD_z$ and is computed on the Binomial $Z$ data.}
\label{fig:hpvWaicEtaVsDelta}
\end{figure}
presents the ELPD values of the candidate $\delta$-SMI (crosses) and $\eta$-SMI distributions (red curve). For each $\delta$ there is an $\eta$ giving the same ELPD. We find a monotone decreasing function transforming the $\eta$ values. The function is chosen so that the ELPD trend across $\eta$ matches that across $\delta$ as closely as possible.

In Figure~\ref{fig:hpvWaicEtaVsDelta} (left), the Bayes posterior gives better posterior predictive performance for the Poisson data, the $Y$'s (largest $ELPD_y$ at small $\delta$) so we choose Bayes when we choose $\delta$ to maximise $ELPD_y$ in the graph on the left.
In this case the $Y$-model is misspecified so the well specified Binomial model ```helps" for $Y$-prediction. In contrast, if we care about predicting the Binomial data $Z$, so we select a $\delta$-SMI posterior using $ELPD_z$, then we see from Figure~\ref{fig:hpvWaicEtaVsDelta} (right) that the Cut model is favored: the $Z$ model is well-specified, so information from the misspecified $Y$-model only worsens performance.

The $\delta$ meta-parameter in $\delta$-SMI seems more readily interpretable than the $\eta$ meta-parameter in $\eta$-SMI. Suppose we use the ELPD and select $\eta^*=0.1$. This seems rather far from Bayes at $\eta=1$. However, based on the $ELPD_y$ values of the Poisson module shown in Fig.~\ref{fig:hpvWaicEtaVsDelta}, $\eta^*=0.1$ gives a similar $ELPD_y$ to $\delta$-SMI with $\delta = 8$. Now typical values of the Poisson data $Y$ are in the hundreds (the median is 162) so ``coarsening" these data with a kernel of bandwidth $\delta=8$ should lead to a mild modification of the posterior.

Many kernels would satisfy Proposition~\ref{prop:ks-smi-interpolates}. 
We investigated sensitivity to the choice of kernel, considering in particular kernels in which the ``bandwidth" $\delta$ was larger at larger $Y$-values (we used the top-hat kernel centred at $y$ with width $\sqrt{y}\delta$).
The results (which we do not report) were robust to this variation at least.



\section{Discussion} \label{sec:discussion}

In this paper we have extended the property of valid belief updates to prequentially additive losses. We gave some examples of prequentially additive losses arising in Cut models and three forms of SMI. These order-coherent inference schemes treat misspecification in models with multiple modules. One criticism of this program is that order-coherence is not axiomatic for misspecified models. However, it seems to us a desirable property if the fitted model imposes conditional independence. 

Another criticism we note in Section~\ref{sec:asymptotics} is that Cut models and $\delta$-SMI do not have correct Frequentist coverage of the pseudo-true parameters in the limit of many observations \citep{Pompe2021}. However, first, we expect SMI to be useful when one module is well specified and we wish to bring in information from other potentially misspecified modules. In our running example, Figure~\ref{fig:toy_multimodular_model}, the Frequentist coverage of the asymptotic Cut-model posterior for $\varphi$ under replication of the $Z$-data will be correct as that module is by assumption well-specified. Secondly, in our experiments in Section~\ref{sec:another-example} on a small data set the distribution of PMSE values obtained for SMI under replication of the data was not worse than Cut and Bayes and often better. This behavior is observed over a range of different levels of misspecification using fitting methods that are available in realistic settings. Finally, the cut-alternative suggested in \cite{Pompe2021}, which does have correct asymptotic Frequentist coverage, is not an order coherent belief update. 

A broader criticism is that the parameters of a strongly misspecified model loose the physical meaning they get from the generative model. Whilst prediction of new data still makes sense, parameter estimation does not. Again, this criticism does not arise when our aim is to control the flow of information from a misspecified model into a well-specified model and estimate parameters in the well specified model.

We used the ELPD as a utility to select a belief update. In general the utility should take into account the objectives of analysis. The ELPD targets predictive performance. When our interest is in parameter estimation and not prediction, we use the ELPD as a proxy for a utility targeting the parameters. We can choose the data on which the ELPD is computed so that the ELPD is sensitive to the parameters we care about. For example, if our aim is to infer $\varphi$ in Figure~\ref{fig:toy_multimodular_model} then $ELPD_z$ in \eqref{eq:ELPD_z} is a natural choice.

The model components identified as ``modules" may to some extent be chosen in the analysis. A module may contain more than one distinct data type, or none. Modules with no data incorporate prior information, and this information may need to be cut, or modulated
in the same way as any other source of information entering the analysis. 
\cite{Styring2017} and \cite{Yu2021variationalcut} give Cut-model analyses, and \cite{carmona20} and \cite{styring22} give $\eta$-SMI analyses of a hierarchical model for archaeological data in which one of the modules has a large vector of missing data, but no observed data. 

We have seen that $\delta$-SMI and $\eta$-SMI can give very similar posteriors, identical in the simple normal example in Section~\ref{sec:biased_data}, and in general depending on the chosen $\delta$-SIM smoothing kernel $K_\delta$. We presented SMI as examples of Gibbs posteriors with loss functions which are only prequentially additive but give valid and order-coherent belief updates. The $\eta$-SMI family of posterior distributions are based on power posteriors. This is a natural choice, but not the only one available. Any order-coherent family of distributions interpolating Cut and Bayes is potentially of interest. One good feature of $\delta$-SMI posteriors is that $\delta$ has the same dimension as the data $Y$,
so the measure of misspecification has a simple interpretation. It is large or small compared to the variation in the sampled $Y$-values. Also, the $\delta$-SMI posterior is a kind of ABC-posterior in which we condition on the data in some neigbourhood of the observed data. However, in $\delta$-SMI this neighborhood is a product space of neighborhoods for each observation $Y_i,\ i=1,...,n$, there is no ``summary statistic" and we recover Bayesian inference as $\delta\to 0$.
This kind of connection between ABC-like methods and misspecification has been noted elsewhere \citep{Miller2018a}. 

One other feature of $\delta$-SMI distinct from $\eta$-SMI is that the likelihood relaxation $p_{\delta}(Y|\varphi,\theta)$ is itself a probability distribution normalised over the data. The power-likelihood $p(Y|\varphi,\theta)^\eta$ in $\eta$-SMI is not normalised in this way. It follows that the imputation distribution $\pi^{(k)}_\delta(\varphi,\tilde\theta|Y,Z)$ is given by Bayes rule for the observation model $p_{\delta}(Y|\varphi,\theta)$. However the inference itself is not Bayesian, unless $\delta=0$, as $\p^{(k)}_\delta(\varphi,\tilde\theta,\theta|Y,Z)$ is not given by Bayes rule.

A number of extensions and variations seem possible. \cite{Goudie2019melding} consider multi-modular models which are in conflict because shared parameters have different priors in different models. They use Markov melding to bring these together in a single model with pooled priors. The pooled priors represent a kind of consensus across modules. This could be combined with SMI if some individual generative models are mispecified. 
In dictatorial pooling the pooled prior is taken to be the prior in one ``authoritative" module. 
This may lead to misspecification in modules sharing the parameter. This is a setting suitable for SMI, where we know which modules are misspecified and need to modulate their influence on inference in the authoritative module.

The Cut and Bayes posteriors can be replaced by distributions derived from the Posterior Bootstrap \citep{Pompe2021} or Bagged posteriors \citep{Huggins2021} and this suggests $\delta$-SMI-like sequences of distributions interpolating these ``Cut" and "Bayes'' distributions by adding ``noise" with bandwidth $\delta$ to $Y$. Since these bootstrapped posterior distributions have good asymptotic Frequentist coverage of the pseudo-true parameters, at least for misspecified variance, it is to hoped that the $\delta$-SMI sequence would inherit these properties.

\newpage



\pagebreak




\bibliographystyle{ba}
\bibliography{references}

\begin{supplement}

\appendix
\renewcommand{\thesection}{A\arabic{section}}

\section{Proofs}

\subsection{Proof of Proposition~\protect\ref{prop:cutmodeladditive}}\label{sec:prop:cutmodeladditive:proof}

{\it Proposition}~\ref{prop:cutmodeladditive}.
  The Cut-model loss $l^{(c)}( \varphi,\theta ; Y, Z ,\pi_0)$ in \eqref{eq:cut_loss} is prequentially additive for the Cut-model posterior, that is, Equation~\eqref{eq:intro_preq_additivity} holds for $l=l^{(c)}$, 
  \[
  \q_k(\varphi, \theta)\propto
  \p^{(c)}(\varphi, \theta|Y^{(1:k)},Z^{(1:k)}), \ k=1,...,K
 \]
  and any partition $Y^{(1:K)},Z^{(1:K)}$ of conditionally independent data $(Y,Z)$. 

\begin{proof}
It is sufficient to show that \eqref{eq:intro_preq_additivity} holds for any partition of the data of size $K=2$ (as we can then split down to any partition), so we split each data set $Y,Z$ into two subsets. Since $\q_0=\pi_0$ and $\q_1$ is the Cut posterior given data $(Y^{(1)},Z^{(1)})$, we should show that
\begin{equation}\label{eq:pacml}
    l^{(c)}( \varphi,\theta ; Y, Z ,\pi_0) = l^{(c)}( \varphi,\theta ;Y^{(1)}, Z^{(1)},\pi_0)+l^{(c)}( \varphi,\theta ; Y^{(2)}, Z^{(2)}, \q_1),
  \end{equation}
where $\q_1(\varphi,\theta)=\p^{(c)}(\varphi,\theta|;Y^{(1)},Z^{(1)})$ so from \eqref{eq:cutpost},
\begin{equation}\label{eq:cutpostY1}
  \q_1(\varphi,\theta)=\pi(\varphi|Z^{(1)})\pi(\theta|Y^{(1)},\varphi).
\end{equation}
The loss on the LHS of \eqref{eq:pacml} is given in \eqref{eq:cut_loss}.
The loss in the first term on the RHS of \eqref{eq:pacml} is
\[
  l^{(c)}( \varphi,\theta ;Y^{(1)}, Z^{(1)},\pi_0)=\l^{(b)}( \varphi,\theta ; Y^{(1)}, Z^{(1)})+\log(p(Y^{(1)}|\varphi))
\]
where
\[
  p(Y^{(1)}|\varphi)=\int p(Y^{(1)}|\varphi, \theta)\pi(\theta|\varphi)d\theta.
\]
The loss in the second term on the RHS of \eqref{eq:pacml} is
\[
  l^{(c)}( \varphi,\theta ; Y^{(2)}, Z^{(2)}, \p^{(c)}(\varphi,\theta|;Y^{(1)},Z^{(1)}))=l^{(b)}( \varphi,\theta ; Y^{(2)}, Z^{(2)})+\log(p(Y^{(2)}|Y^{(1)},\varphi))
\]
where
\[
  p(Y^{(2)}|Y^{(1)},\varphi)=\int p(Y^{(2)}|\varphi, \theta)\q_1(\theta|\varphi)d\theta,
\]
since $\q_1(\theta|\varphi)$ is the ``prior'' passed on to the second stage from the belief update in the first stage. From \eqref{eq:cutpostY1}, this is the conditional probability for $\theta|Y^{(1)},\varphi$ in the Cut-model posterior $\p^{(c)}$ so $\q_1(\theta|\varphi)=\pi(\theta|Y^{(1)},\varphi)$.

We are checking the condition holds for iid data so for the Bayes loss,
\[
  \l^{(b)}( \varphi,\theta ; Y, Z)=\l^{(b)}( \varphi,\theta ; Y^{(1)}, Z^{(1)})+\l^{(b)}( \varphi,\theta ; Y^{(2)}, Z^{(2)}).
\]
Assembling the terms,
\begin{align*}
  RHS \protect \eqref{eq:pacml}
   & =
  \l^{(b)}( \varphi,\theta ; Y, Z)+\log p(Y^{(1)}|\varphi) +\log \int p(Y^{(2)}|\varphi, \theta)\pi(\theta|Y^{(1)},\varphi)d\theta                          \\
   & =\l^{(b)}+\log p(Y^{(1)}|\varphi) +\log \int p(Y^{(2)}|\varphi, \theta)\frac{p(Y^{(1)}|\varphi, \theta)\pi(\theta|\varphi)}{p(Y^{(1)}|\varphi)}d\theta \\
   & =\l^{(b)}+\log p(Y|\varphi),
\end{align*}
since $p(Y^{(2)}|\varphi,\theta)p(Y^{(1)}|\varphi, \theta)=p(Y|\varphi, \theta)$ and $p(Y|\varphi)$ is given by \eqref{eq:marginal_Y_phi}. This completes the proof as
\[
  \l^{(b)}+\log p(Y|\varphi)=l^{(c)}( \varphi,\theta ; Y, Z ,\pi_0).
\]
\end{proof}

\subsection{Proof of Theorem~\protect\ref{thm:preq_add_coherent_gives_post}}\label{sec:preq_add_coherent_gives_post:proof}

{\it Theorem~\ref{thm:preq_add_coherent_gives_post}.} If a loss $l$ is prequentially additive with respect to the belief update given by the Gibbs posterior, 
  \[
  \psi^{(q)}(l(\varphi, \theta;Y,Z,\pi_0),\pi_0)\propto \exp(-l(\varphi, \theta;Y,Z,\pi_0)) \pi_0(\varphi, \theta)
  \]
  then $\psi^{(q)}$ is order-coherent. It further holds that $L(\nu; Y,Z,\pi_0)$ in \eqref{eq:big_post_loss}
  is the only valid loss for an order-coherent belief update and $\psi^{(q)}$ itself is the optimal valid order-coherent belief update $\psi$ in \eqref{eq:valid_belief_update}.

\begin{proof}
The candidate $\psi^{(q)}$ is order-coherent due to the exponential form. We have from Definition~\ref{def:preq-order-coherent} and \eqref{eq:psi-q-in-valid-thm} that
\begin{align*}
    \psi^{(q)}\{ l(\varphi, \theta;Y^{(2)},Z^{(2)},\q_1), \q_1 \}&\propto \exp(-l(\varphi, \theta;Y^{(2)},Z^{(2)},\q_1)) \q_1(\varphi, \theta),\\
    \intertext{then expanding the $\q_1$ ``prior'' using \eqref{eq:p1-in-def-order-co} and \eqref{eq:psi-q-in-valid-thm} again,} 
    &\propto \exp(-l(\varphi, \theta;Y^{(2)},Z^{(2)},\q_1)) \exp(-l(\varphi, \theta;Y^{(1)},Z^{(1)},\pi_0)) \pi_0(\varphi, \theta),\\
    \intertext{but $l$ is prequentially additive with respect to $\psi^{(q)}$ so,}
    &\propto \exp(-l(\varphi, \theta;Y,Z,\pi_0))    \pi_0(\varphi, \theta),
\end{align*}
which we recognise as $\psi\{ -l(\varphi, \theta;Y,Z,\pi_0), \pi_0 \}$. This verifies that \eqref{eq:preq-order-coherent} holds.

We now show that $L(\nu; Y,Z,\pi_0)$ in \eqref{eq:big_post_loss} is the only valid loss in the sense of \cite{Bissiri2016}. Our proof  shows that we can substitute prequential additivity for additivity in the Theorem in the supplement to \cite{Bissiri2016} which establishes KL as the unique prior to posterior loss in \eqref{eq:big_post_loss}, so the following is very similar. 

Let $\xi=(\varphi,\theta)$ and $O=(Y,Z)$ so the belief update is from $\pi_0(d\xi)$ to $\nu(d\xi)$ under the loss $l(\xi;O,\pi_0)$.
Denote by $\xi \in \Omega$ the parameter space of $(\varphi,\theta)$. We assume the total loss must be the sum of the expected loss and a prior to posterior divergence $D_g$, that is, 
\[
L(\nu;O,\pi_0)=\mathbb{E}_\nu[l(\xi;O,\pi_0)]+D_g(\nu,\pi_0).
\] 
\cite{Bissiri2016} justify this form which we take as given. They establish the valid belief update for the class of $g$-divergences,
\[
D_g(\nu,\pi_0)=\int g\!\left(\frac{d\nu}{d\pi_0}\right) \pi_0(d\xi)
\] 
with $g$ a fixed differentiable and convex function from $(0,\infty)$ to $\R$ satisfying $g(1) = 0$. Under these conditions they give a concise proof that $D_g$ must be the KL divergence (in fact, $g(x)=k x\log x + (g'(1)-k)(x-1)$ for some $k>0$ - the extra terms integrate to zero). The authors cite \cite{bissiriwalker12a} for a proof under weaker conditions. They show that over this class of $g$-divergences, the KL divergence is necessary and sufficient for the optimal belief update $\psi$ to be order-coherent for every parameter space $\Omega$ and every loss such that the objects involved exist.

First of all is clear that KL is sufficient for order-coherence as the optimal valid belief update is then equal to the Gibbs posterior $\psi^{(q)}$ (see below) and we have seen this is order-coherent under the conditions of Theorem~\ref{thm:preq_add_coherent_gives_post}. In order to show KL is necessary it is enough to give an example where the KL divergence is the only $g$-divergence giving coherence, so \cite{Bissiri2016} take a parameter space with just two states, $\Omega=\{\xi_1,\xi_2\}$ say. Let $O^{(1)}=(Y^{(1)},Z^{(1)})$, $O^{(2)}=(Y^{(2)},Z^{(2)})$ and  \[\q_1(\xi)\propto \exp(-l(\xi;O^{(1)},\pi_0))\pi_0(\xi)\]
in Definitions~\ref{def:preq_add_first} and \ref{def:preq-order-coherent}, using the belief update $\psi^{(q)}$ which makes $l(\xi;O,\q)$ prequentially additive. Now take $I_1=(O^{(1)},\pi_0)$, $I_2=(O^{(2)},\q_1)$ and $I=(O,\pi_0)$ in the proof page 3 of the supplement to \cite{Bissiri2016}. We go through this to make it clear that everthing continues to fall into place and the presence of $\q_1$ inside the information $I_2$ is just what we need to make things work. We should keep in mind below that that $\pi_0$ and $\q_1$ are fixed pieces of information inside $I_1$ and $I_2$ as $p$ is varied.

By prequential additivity, 
\begin{equation}\label{eq:pre_add_in_bissiri_proof}
    l(\xi;I)=l(\xi;I_1)+l(\xi;I_2).
\end{equation}
As there are just two states the distribution $\nu$ is a probability mass function of the form $p\I_{\xi=\xi_1}+(1-p)\I_{\xi=\xi_2}$ parameterised by $p$ so we substitute $p$ for $\nu$ and write $L(p;I,\pi_0)$. The prior is
\[
\pi_0(\xi)=p_0\I_{\xi=\xi_1}+(1-p_0)\I_{\xi=\xi_2},
\]
for some $p_0\in [0,1]$. Let 
\begin{equation}\label{eq:p1_def_bissiri_proof}
    p_1=\arg\min_{p\in [0,1]} L(p;I_1,\pi_0)
\end{equation}
so that the belief update from the prior, with the first block of data, is
\[
\p_1(\xi)=p_1\I_{\xi=\xi_1}+(1-p_1)\I_{\xi=\xi_2}.
\]
The overall belief update $\p_2$ is
given by the belief update from the prior with all of the data, so we set
\begin{equation}\label{eq:p2_def_bissiri_proof}
    p_2=\arg\min_{p\in [0,1]} L(p;I,\pi_0).
\end{equation} 
and define
\[
\p_2(\xi)=p_2\I_{\xi=\xi_1}+(1-p_2)\I_{\xi=\xi_2}.
\]
The requirement that the belief update be order-coherent imposes
\begin{equation}\label{eq:order_coherence_in_bissiri_proof}
\arg\min_{p\in [0,1]} L(p;I_2,\p_1)=p_2.
\end{equation}
With these identifications for $I_1, I_2$ and $I$ and substituting $p_0, p_1$ and $p_2$ for $\pi_0, \p_1$ and $\p_2$ in the notation, the losses can be written
\begin{align}
    L(p,I_1,p_0)&=p \;  l(\xi_1;I_1) + (1-p) \;  l(\xi_2;I_1) +
    p\;  g\!\left(\frac{p}{p_0}\right) + (1-p)\;  g\!\left(\frac{1-p}{1-p_0}\right),\\
    L(p,I_2,p_1)&=p \;  l(\xi_1;I_2) + (1-p) \;  l(\xi_2;I_2) +
    p\;  g\!\left(\frac{p}{p_1}\right) + (1-p)\;  g\!\left(\frac{1-p}{1-p_1}\right),\\
    L(p,I,p_0)&=p \;  l(\xi_1;I) + (1-p) \;  l(\xi_2;I) +
    p\;  g\!\left(\frac{p}{p_0}\right) + (1-p)\;  g\!\left(\frac{1-p}{1-p_0}\right).
\end{align}
Differentiating with respect to $p$ in order to solve Equations~\eqref{eq:p1_def_bissiri_proof}, \eqref{eq:p2_def_bissiri_proof} and \eqref{eq:order_coherence_in_bissiri_proof} respectively gives
\begin{align}
     l(\xi_1;I_1) - l(\xi_2;I_1) &=
    g'\!\left(\frac{p_1}{p_0}\right) - g'\!\left(\frac{1-p_1}{1-p_0}\right),\label{eq:g-dash1}\\
    l(\xi_1;I) - l(\xi_2;I) &=
    g'\!\left(\frac{p_2}{p_0}\right) - g'\!\left(\frac{1-p_2}{1-p_0}\right),\label{eq:g-dash-all}\\
    l(\xi_1;I_2) - l(\xi_2;I_2) &=
    g'\!\left(\frac{p_2}{p_1}\right) - g'\!\left(\frac{1-p_2}{1-p_1}\right).\label{eq:g-dash2}
\end{align}
Now by prequential additivity in \eqref{eq:pre_add_in_bissiri_proof} the sum of
the LHS of \eqref{eq:g-dash1} and the LHS of \eqref{eq:g-dash2} is
\[
l(\xi_1;I_1) - l(\xi_2;I_1)+l(\xi_1;I_2) - l(\xi_2;I_2)=l(\xi_1;I) - l(\xi_2;I)
\]
which is the LHS of \eqref{eq:g-dash-all} so
\begin{equation}
   g'\!\left(\frac{p_1}{p_0}\right) - g\!\left(\frac{1-p_1}{1-p_0}\right)+ g'\!\left(\frac{p_2}{p_1}\right) - g'\!\left(\frac{1-p_2}{1-p_1}\right) = g'\!\left(\frac{p_2}{p_0}\right) - g'\!\left(\frac{1-p_2}{1-p_0}\right).
\end{equation}
This is Equation~7 in
the proof on pages 3 and 4 of the supplement to \cite{Bissiri2016}. From this point the proof goes through without change: the assumed and derived properties of $g$ require $g(x)=k x\log x + (g'(1)-k)(x-1)$ and the second term doesn't contribute to $D_g$ as it integrates to zero. 


This establishes $L$ in \eqref{eq:big_post_loss} as a valid loss for an order-coherent belief update. However, $\pi_0$ is fixed in the variation over $\nu$ in \eqref{eq:valid_belief_update}, so \eqref{eq:psi-q-in-valid-thm} is the density of the measure $\nu$ maximising \eqref{eq:big_post_loss} for both additive and prequentially additive losses. The proof of this step is unchanged from that given at the end of Section 1.1 of \cite{Bissiri2016}.
\end{proof}

\subsection{Proof of Proposition~\protect\ref{prop:gam-eta-del_are_preq_add}}\label{sec:prop:gam-eta-del_are_preq_add:proof}

{\it Proposition}~\ref{prop:gam-eta-del_are_preq_add}.
The loss functions for $\gamma$-SMI, $\eta$-SMI and $\delta$-SMI given respectively in \eqref{eq:tmsmi_loss}, \eqref{eq:smi_loss} and \eqref{eq:kssmi_loss} are prequentially additive with respect to the belief updates given respectively in \eqref{eq:gamma_smi_post_tmp}, \eqref{eq:eta_smi_post_pow} and \eqref{eq:kssmi_def}.

\begin{proof}
The proof is similar to that of Proposition~\ref{prop:cutmodeladditive}.
We wish to show that
  \begin{equation}\label{eq:pacml2}
    l( \varphi,\theta ; Y, Z ,\pi_0) = l( \varphi,\theta ;Y^{(1)}, Z^{(1)},\pi_0)+l( \varphi,\theta ; Y^{(2)}, Z^{(2)}, \p(\varphi,\theta|;Y^{(1)},Z^{(1)}))
  \end{equation}
 for $(l,\p)$ in turn $(l^{(t)},\p^{(t)}_\gamma)$, $(l^{(s)},\p^{(s)}_\eta)$ and $(l^{(k)},\p^{(k)}_\delta)$. 
 
The respective total losses appearing on the LHS of \eqref{eq:pacml2} are given in \eqref{eq:tmsmi_loss}, \eqref{eq:smi_loss} and \eqref{eq:kssmi_loss}.
The corresponding losses in the first term on the RHS of \eqref{eq:pacml2} are respectively
\begin{align*}
    l^{(t)}( \varphi,\theta ;Y^{(1)}, Z^{(1)},\pi_0)&=\l^{(b)}( \varphi,\theta ; Y^{(1)}, Z^{(1)})+(1-\gamma)\log p(Y^{(1)}|\varphi)\\
    l^{(s)}( \varphi,\theta ;Y^{(1)}, Z^{(1)},\pi_0)&=\l^{(b)}( \varphi,\theta ; Y^{(1)}, Z^{(1)})-\eta \log p(Y^{(1)}|\varphi,\theta)+\log p(Y^{(1)}|\varphi)\\
    l^{(k)}( \varphi,\theta ;Y^{(1)}, Z^{(1)},\pi_0)&=\l^{(b)}( \varphi,\theta ; Y^{(1)}, Z^{(1)})- \log p_{\delta}(Y^{(1)}|\varphi,\theta)+\log p(Y^{(1)}|\varphi)
\end{align*}
where
\[
  p(Y^{(1)}|\varphi)=\int p(Y^{(1)}|\varphi, \theta)\pi(\theta|\varphi)d\theta
\]
throughout.
The loss in the second term on the RHS of \eqref{eq:pacml2} is
respectively
\begin{align*}
    l^{(t)}( \varphi,\theta ;Y^{(2)}, Z^{(2)},\p^{(t)}_1)&=\l^{(b)}( \varphi,\theta ; Y^{(2)}, Z^{(2)})+(1-\gamma)\log p(Y^{(2)}|Y^{(1)},\varphi) \\
    l^{(s)}( \varphi,\theta ;Y^{(1)}, Z^{(1)},\pi_0)&=\l^{(b)}( \varphi,\theta ; Y^{(1)}, Z^{(1)})-\eta \log p(Y^{(2)}|\varphi,\theta)+\log p(Y^{(2)}|Y^{(1)},\varphi)\\
    l^{(k)}( \varphi,\theta ;Y^{(1)}, Z^{(1)},\pi_0)&=\l^{(b)}( \varphi,\theta ; Y^{(1)}, Z^{(1)})- \log p_{\delta}(Y^{(2)}|\varphi,\theta)+\log p(Y^{(2)}|Y^{(1)},\varphi)
\end{align*}
where again
\[
  p(Y^{(2)}|Y^{(1)},\varphi)=\int p(Y^{(2)}|\varphi, \theta)\pi(\theta|Y^{(1)},\varphi)d\theta.
\]
This last relation holds throughout because $\pi(\theta|Y^{(1)},\varphi)$ is the ``prior'' passed on to the second stage from the belief update in the first stage, and the conditional probability for $\theta|Y^{(1)},\varphi$ is actually the same in $\gamma$-SMI, $\eta$-SMI and $\delta$-SMI as can be seen by inspecting \eqref{eq:gamma_smi_marg_post_phitheta}, \eqref{eq:eta_smi_marg_post} and \eqref{eq:kssmi_marginal_phitheta}.
 
The contribution from the Bayes-loss $l^{(b)}$ is straightforwardly additive.
Assembling the terms,
\begin{align*}
  RHS \protect \eqref{eq:pacml2}
   & =
  l( \varphi,\theta ;Y^{(1)}, Z^{(1)},\pi_0 )+l( \varphi,\theta ; Y^{(2)}, Z^{(2)},   \p(\varphi,\theta|;Y^{(1)},Z^{(1)}))                                  \\
   & = \l^{(b)}( \varphi,\theta ; Y, Z)\ 
  \left\{\begin{array}{l}
       +\  (1-\gamma)\left[\log p(Y^{(1)}|\varphi)+\log p(Y^{(2)}|Y^{(1)},\varphi)\right]  \\
       -\ \eta \left[\log p(Y^{(1)}|\varphi,\theta)+\log       p(Y^{(2)}|\varphi,\theta)\right]+\log p(Y^{(1)}|\varphi)+\log p(Y^{(2)}|Y^{(1)},\varphi) \\
       -\ \log p_\delta(Y^{(1)}|\varphi,\theta)-\log p_\delta(Y^{(2)}|\varphi,\theta) +log p(Y^{(1)}|\varphi)+\log p(Y^{(2)}|Y^{(1)},\varphi) 
  \end{array} \right.                      \\
  & = 
  \l^{(b)}( \varphi,\theta ; Y, Z)\ 
  \left\{\begin{array}{ll}
       +\ (1-\gamma) \log p(Y|\varphi)  & \mbox{($\gamma$-SMI)}\\
       -\ \eta  \log p(Y|\varphi,\theta) + \log p(Y|\varphi) & \mbox{($\eta$-SMI)}\\
       -\ \log p_\delta(Y|\varphi,\theta) +\log p(Y|\varphi) & \mbox{($\delta$-SMI)}
  \end{array} \right.                      \\
   & = 
  \left\{\begin{array}{c}
        l^{(t)}( \varphi,\theta ;Y, Z,\pi_0 )\\
        l^{(s)}( \varphi,\theta ;Y, Z,\pi_0 )\\
         l^{(k)}( \varphi,\theta ;Y, Z,\pi_0 )
  \end{array} \right.                      
\end{align*}
where we used the defining equations \eqref{eq:tmsmi_loss}, \eqref{eq:smi_loss} and \eqref{eq:kssmi_loss} for the losses for the full data to make the last step and
\begin{align*}
    p(Y^{(2)}|\varphi,\theta)p(Y^{(1)}|\varphi, \theta)=p(Y|\varphi, \theta)\\
    p_\delta(Y^{(2)}|\varphi,\theta)p_\delta(Y^{(1)}|\varphi, \theta)=p_\delta(Y|\varphi, \theta)\\
  p(Y|\varphi)=p(Y^{(1)}|\varphi)p(Y^{(2)}|Y^{(1)},\varphi)  
\end{align*}
for these iid data.
\end{proof}

\section{Details of examples}

\subsection{Simulation study: Biased data}
\label{app:biased_data}



The model components are given in Section~\ref{sec:biased_data}. The fitted and true observaiton models for $Y$ and $Z$ are the same.
We used parameter values $n=50, m=25, \sigma_\theta=0.33, \sigma_y=1, \sigma_z=3, \varphi^*=0, \theta^*=1$.

We estimated the PMSE's from samples. The posterior predictive density for $y,z\in \R$ is needed in order to calculate the ELPD. It is,
\begin{align*}
  p^{(k)}_{y,z,\delta}(y,z\mid Y,Z)=N((y,z)^T;\mu_{yz},\Sigma_{yz}),
\end{align*}
with
\[
  \mu_{yz}=\left(\begin{array}{c} (1-\rho)\mu_\delta+\rho\bar Y \\
      \mu_\delta
    \end{array}\right),\quad
  \Sigma_{yz}=\left(\begin{array}{cc}
      (1-\rho)\sigma_\delta^2+\sigma_{\theta|Y,\varphi}^2+\sigma_y^2 & (1-\rho)\sigma_\delta^2    \\
      (1-\rho)\sigma_\delta^2                                        & \sigma_\delta^2+\sigma_z^2
    \end{array}\right).
\]
The formula for the ELPD in this setting is
\[
  ELPD_{y,z}(Y,Z;\delta)=-\log(2\pi)-\frac{1}{2}\log(\det(\Sigma_{yz}))-\frac{1}{2}E_{p^*}\left(((y,z)-\mu_{yz}^T)\Sigma^{-1}_{yz}((y,z)^T-\mu_{yz})\right)
\]
where the expectation in $y,z$ is taken in the true generative distribution $p^*$ where
\[
  y\sim N(\theta^*+\varphi^*,\sigma_y^2),\quad z\sim N(\varphi^*,\sigma^2_z).
\]
Let $\mu^*=(\theta^*+\varphi^*,\varphi^*)$ and $\Sigma^*=\mbox{diag}(\sigma_y^2,\sigma_z^2)$.
The expectation is
\[
E_{p^*}\left(((y,z)-\mu_{yz}^T)\Sigma^{-1}_{yz}((y,z)^T-\mu_{yz})\right)=\mbox{trace}[\Sigma^{-1}_{yz}\Sigma^*]+ (\mu^{*T}-\mu_{yz}^T)\Sigma^{-1}_{yz}(\mu^*-\mu_{yz}).
\]



\subsection{Simulation study: Regression data}\label{app:regression_example}

The fitted and true models are given in Section~\ref{sec:another-example}. The parameter settings are $n=50$, $m=50$, $\sigma_y=0.25$, $\sigma_z=3$ and $X_i\sim U(0,2),\ i=1,...,n$.
The true parameter values are $\varphi^*=0$ and $\theta^*=1$.

\end{supplement}

\end{document}